\documentclass[12pt]{article}
\usepackage{times}

\usepackage{mystyle} 
\newtheorem{lemma}[theorem]{Lemma}
\renewcommand{\cal}[1]{\mathcal{#1}}

\usepackage{cite}
\usepackage{tcolorbox}
\usepackage{graphicx}
\usepackage{url}
\usepackage[stable]{footmisc}
\usepackage{hyperref}

\usepackage[linesnumbered,ruled,vlined ]{algorithm2e}

\title{Utility-based Resource Allocation and Pricing for Serverless Computing}
\author{
Vipul Gupta$^*$, Soham Phade$^*$, Thomas Courtade, Kannan Ramchandran 
\blfootnote{Equal contribution. This work was supported by NSF Grants CCF-2007669, CCF-1704967, CCF-0939370, CIF-1703678, CNS-1527846 and  CCF-1618145.}
\\
{\normalsize Department of EECS,  University of California, Berkeley}\\
{\normalsize Email: vipul\_gupta, soham\_phade, courtade, kannanr@eecs.berkeley.edu}
}
\date{\vspace{-5mm}}

\begin{document}
\maketitle

\begin{abstract}

Serverless computing platforms currently rely on basic pricing schemes that are static and do not reflect customer feedback. This leads to significant inefficiencies from a total utility perspective.  As one of the fastest-growing cloud services, serverless computing provides an opportunity to better serve both users and providers through the incorporation of market-based strategies for pricing and resource allocation. With the help of utility functions to model the delay-sensitivity of customers, we propose a novel scheduler to allocate resources for serverless computing. The resulting resource allocation scheme is optimal in the sense that it maximizes the aggregate utility of all users across the system, thus maximizing social welfare. Our approach gives rise to a natural dynamic pricing scheme that is obtained by solving an optimization problem in its dual form. We further develop feedback mechanisms that allow the cloud provider to converge to  optimal resource allocation, even when the users' utilities are private and unknown to the service provider. Simulations show that our approach can track market demand and achieve significantly higher social welfare (or, equivalently, cost savings for customers) compared to existing schemes.
\end{abstract}



\section{Introduction}

In this paper, we provide a 
novel framework
for resource allocation in cloud computing systems.
We focus on the serverless computing setting that has recently garnered significant attention from industry (e.g., Amazon Web Services (AWS) Lambda, Microsoft Azure Functions, Google Cloud Functions) as well as the systems community (see, e.g.,~\cite{ serverless_computing, pywren,numpywren, hellerstein2018serverless, berkeley_view, schleier2021serverless}).
Serverless computing is especially appealing due to its \emph{ease of management} and \emph{elasticity}, where the typical jobs are comprised of several independent functions to be executed. These functions have relatively small system requirements such as execution time and memory storage \cite{berkeley_view, pywren}.

The pricing mechanism (and the corresponding resource allocation scheme employed by the cloud service provider) is one of the most important tools in influencing the usage of cloud resources.
A typical customer's goal is to obtain the highest quality of service (QoS) for a reasonable and affordable price.
Thus, how the service provider allocates resources and charges its customers affects customer behavior, loyalty to the provider, and ultimately its success.
Current popular cloud computing providers--such as AWS, Microsoft Azure, and Google Cloud--employ a pricing scheme 
referred to as ``pay-per-use fixed pricing.''
It scales linearly in the resources utilized, such as time, memory, and the number of jobs, regardless of the nature (e.g., delay-sensitive or not?), and importance (e.g., critical or not?) of the users' applications.

This pricing scheme needs to be revisited, in part due to the following reasons.
\begin{itemize}
    \item {\bf Pricing based on market-demand}: 
    The growing demand implies an increasingly infeasible task for cloud providers to schedule all jobs instantly \cite{ms_azure_2020}.
    An important feature of serverless systems is the elimination of up-front commitment by users. 
    These jobs are typically triggered by external events that are not in control of the cloud service provider (e.g., receipt of a message); see \cite{ms_azure_2020} for examples of different triggers. 
    As a result, the demand for resources can vary significantly over time.  
    Hence, the prices should adapt to this changing demand in real-time. 
    For example, this will ensure that the cloud provider will serve the users who most value their jobs during surge periods, and are therefore willing to pay premium prices.
    \item {\bf Resource allocation and pricing based on delay sensitivity}: In addition to dynamic pricing based on demand, 
it is important to set resource allocation and pricing for jobs based on their \emph{delay-sensitivity} and allocate resources accordingly.
Job completion times form an integral part on service level agreements that governs the quality of service of the users \cite{patel2009service,wilson2011better}.
Users have heterogeneous jobs and have different requirements with respect to their service delay.
For example, \emph{urgent} jobs (that need to be executed in real-time, such as model deployment \cite{serving_models_serverless} and real-time video compression \cite{serverless_video_encoding}) may need prioritization, whereas \emph{enduring} jobs (that can be put into queues with reasonable wait-times, such as optimization for machine learning \cite{training_DNN_serverless,osn,cirrus,serverless_ml} and scientific computing \cite{serverless_faaster_better_cheaper, numpywren,oversketch}) could be put on hold.  It is therefore desirable for 
the pricing scheme to provide appropriate incentives to users. 
For example,  premium rates could be applied to urgent jobs, and discounts applied to  enduring jobs.
\end{itemize}



To address these issues, 
we develop a \emph{dynamic multi-tier pricing scheme} that incentivizes users to bid optimally for resources that are tailored to their requirements and delay-sensitivity characteristics. 
To articulate the notion of demand and delay-sensitivity, we  adopt the concept of utility functions from economics.
We consider user utilities as a function of delays in job completion times.
This enables us to naturally differentiate jobs based on their delay-sensitivity characteristics and allocate the resources optimally.
Some examples of such utility functions are shown in Fig. \ref{fig:utility_examples}. 
%
Utility, being an abstract concept, it is common to consider the relatively more tangible notion of \emph{willingness to pay} in lieu of the utility function.
For example, if a user is willing to pay a maximum of $\$1$ for her job to get completed with a delay of 1 minute, then we will say that her utility function takes value $\$1$ at $t = 1$ minute. 
Thus, in Fig. \ref{fig:utility_examples}, the y-axis can also be thought of as the amount the users are willing to pay w.r.t. job completion times.

\begin{figure}[t]
  \centering
  \includegraphics[width=0.9\linewidth]{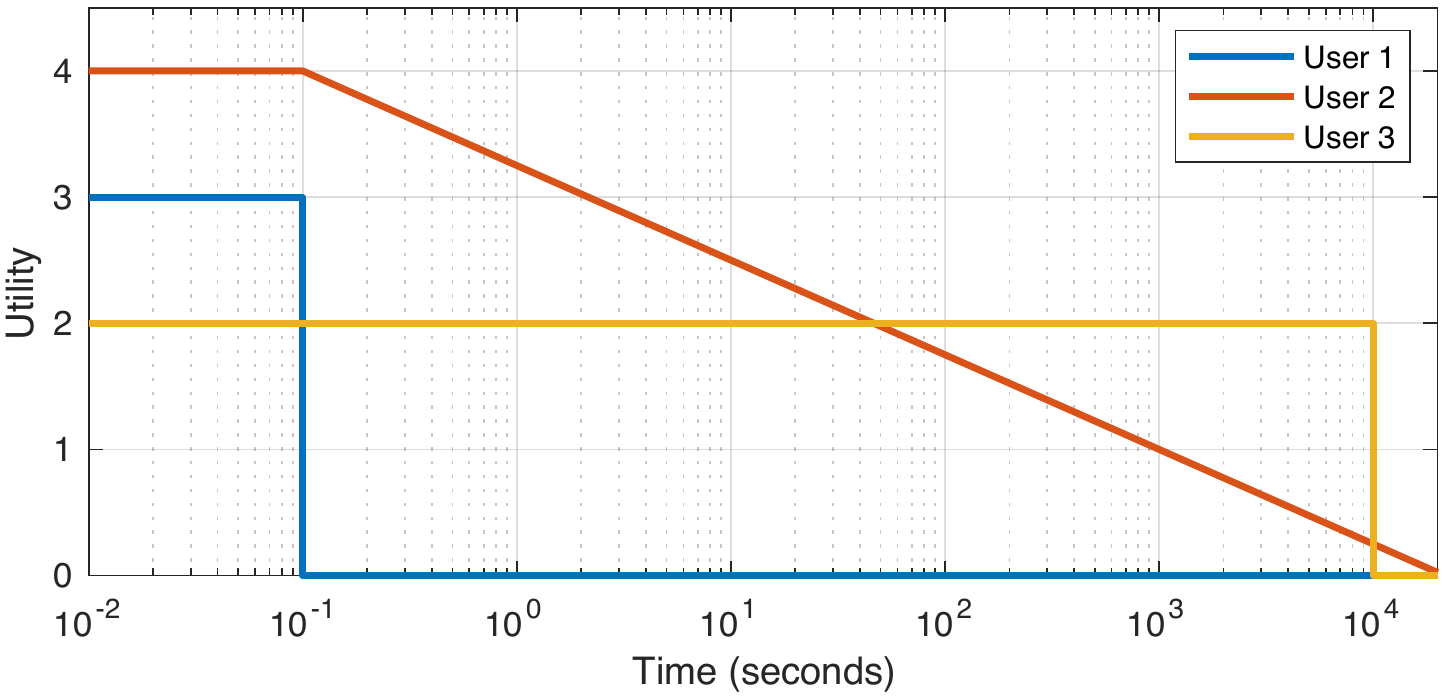}
  \label{fig:naive1}
\caption{Examples of utility functions for 3 users that depend on completion time. User 1 obtains utility only when her job gets completed under 0.1 seconds (e.g., machine learning inference). User 2 obtains diminishing returns as time passes (e.g., market-data analysis for algorithmic trading). User 3, on the other hand, does not care as long as her job is completed within 10000 seconds (e.g., large-scale machine learning).  }
\label{fig:utility_examples}
\end{figure}

\begin{figure*}[t]
 \centering
 \includegraphics[width=0.7\textwidth]{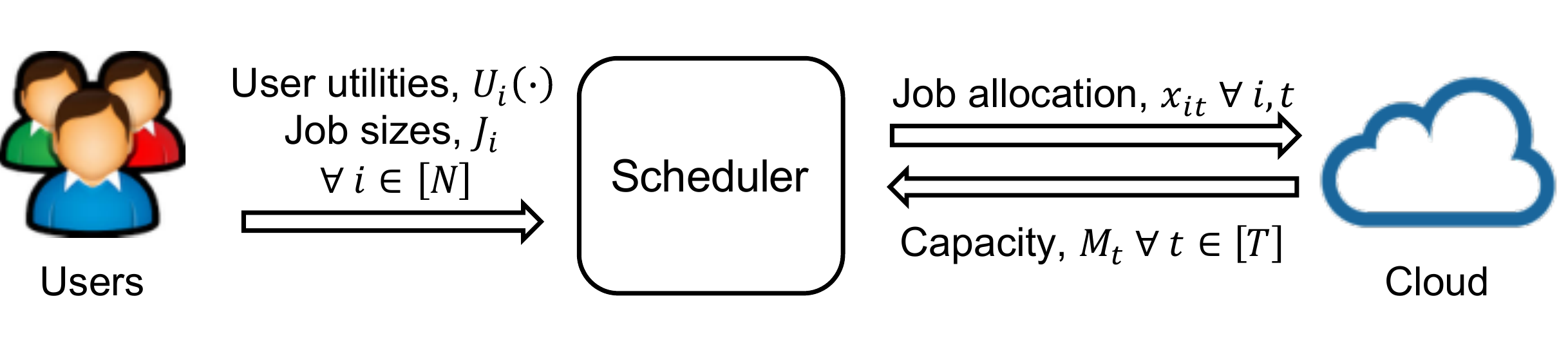}
 \caption{A block diagram schematic of the scheduler. It takes as its input the job size, $J_i$, and the utility function, $U_i(\cdot)$, for all users $i =1,\cdots, N$, the capacity constraints or machine availability information from the cloud, and outputs a job allocation schedule (notation is described in Sec. \ref{sec: prob_form}).
Here, the job sizes capture the demand for resources, the utility functions capture the delay-sensitivity of the users, and the capacity constraints capture the supply of resources.
}
 \label{fig:sys_prob}
 \end{figure*}

As a notion of optimal allocation, we set for ourselves the goal of scheduling jobs so that the total utility gained by the users in the system is maximized.
Social welfare as a concept is very important to many companies that prioritize customer satisfaction\footnote{For example, see https://www.amazon.jobs/en/principles}.
We formulate this as an optimization problem that maximizes the social welfare (i.e., the sum of the utilities received by all the users).
The resulting scheduling problem is inherently a dynamic optimization problem.
To incorporate utilities into this problem, we first identify a static version of the scheduler problem that can be run periodically.
We elaborate on this in Sec.~\ref{sec: prob_form}.
A prototype for our static scheduler is shown in Figure~\ref{fig:sys_prob}.


\begin{figure*}[ht]
 \centering
 \includegraphics[width=0.7\textwidth]{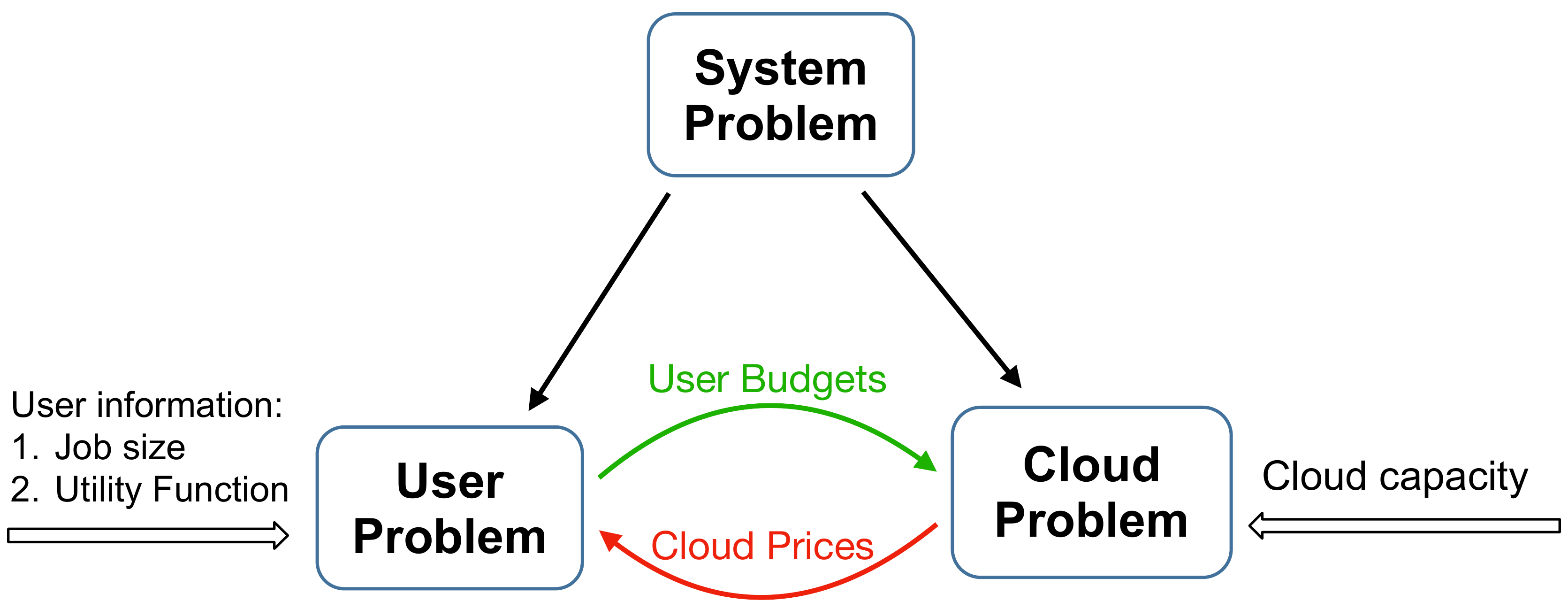}
 \caption{Decomposing the system problem into a user and cloud problems. The users and cloud interact iteratively to solve the system problem that maximizes net utility, thus, obtaining optimal pricing and resource allocation scheme. 
}
 \label{fig:sys_prob_decomp}
 \end{figure*}

It is not enough to solve this optimization problem to be of practical use because
the utility (or willingness-to-pay) of a user is unknown to the service provider.
We cannot simply assume that the users would truthfully report their utility functions to the service provider.
Even if the users are willing to do so, the high dimensional utility functions that could change with time could be hard to communicate in real-time.
In fact, oftentimes, the users themselves are unaware of their utility functions. 
Even if we consider a low dimensional family of typical utility functions, which could potentially resolve the problem of communication, there exist privacy concerns under which users are reluctant to report their utility functions.
Thus, it is impractical to assume, in general, that the cloud service provider has access to the users' utility function.
We need to find an indirect way based on market mechanisms to resolve this problem.

A similar problem was faced by the network community in allocating bandwidth over the Internet using TCP \cite{mackie1995pricing, mo2000fair, low1999optimization, kunniyur2003end}.
Inspired by the seminal works in \cite{kelly1997charging, kelly1998rate}, we wish to decompose the scheduler optimization problem into user problems---one for each user---and a cloud problem.
Each user problem interacts with the cloud problem through the prices published by the cloud provider and the corresponding budget responses from the users (see Fig. \ref{fig:sys_prob_decomp} for an illustration).
A user's budget response is a function of her utility function and the prices published by the cloud.
This inserts a filter between the utility functions of the users and their reports to the network providing them a layer of privacy and making the communication between the cloud and the users feasible.%
\footnote{Although an interesting feature of our algorithm, we do not pursue the privacy aspect in this paper and leave it for future work.}
The above decomposition allows us to track the optimal performance with limited feedback from users in the form of their budgets (which is a response to the currently advertised prices).

Our problem can be viewed as the scheduling analog of network bandwidth allocation problem.
However, the scheduling optimization problem by itself is not amenable for this decomposition. 
The underlying reason for this can roughly be stated as follows: Under TCP, the users obtain utility based on their allocated bandwidth. 
The network imposes restrictions on the total bandwidth that can be allocated over each of the links in the network.
In our case, utility is awarded to users when they complete their job. 
The completion time in turn depends on the number of cloud computing resources allocated to this user and when they are assigned.
Thus there is a non-trivial relation between the allocated resources and the outcomes that the users care about (i.e. the domain of their utilities a function).
This makes the problem significantly more complicated.
In fact, we note that the original problem is NP-hard.
We provide a relaxation that is tractable.
This relaxation can be viewed as a vector version of the network optimization problem as proposed in \cite{kelly1997charging}.
We then identify a sparsity property in the solution to this problem and leverage it to show that the obtained solution is a good approximation to the original intractable problem.
We provide a provable bound on the approximation gap.

To the best of our knowledge, this is the first paper to extend the network resource allocation problem to the scheduling setting in a decentralized way.
The ideas surrounding relaxation and the use of sparsity in the resulting problem could help in many avenues beyond serverless computing.

\begin{figure*}[ht]
 \centering
 \includegraphics[width=0.7\textwidth]{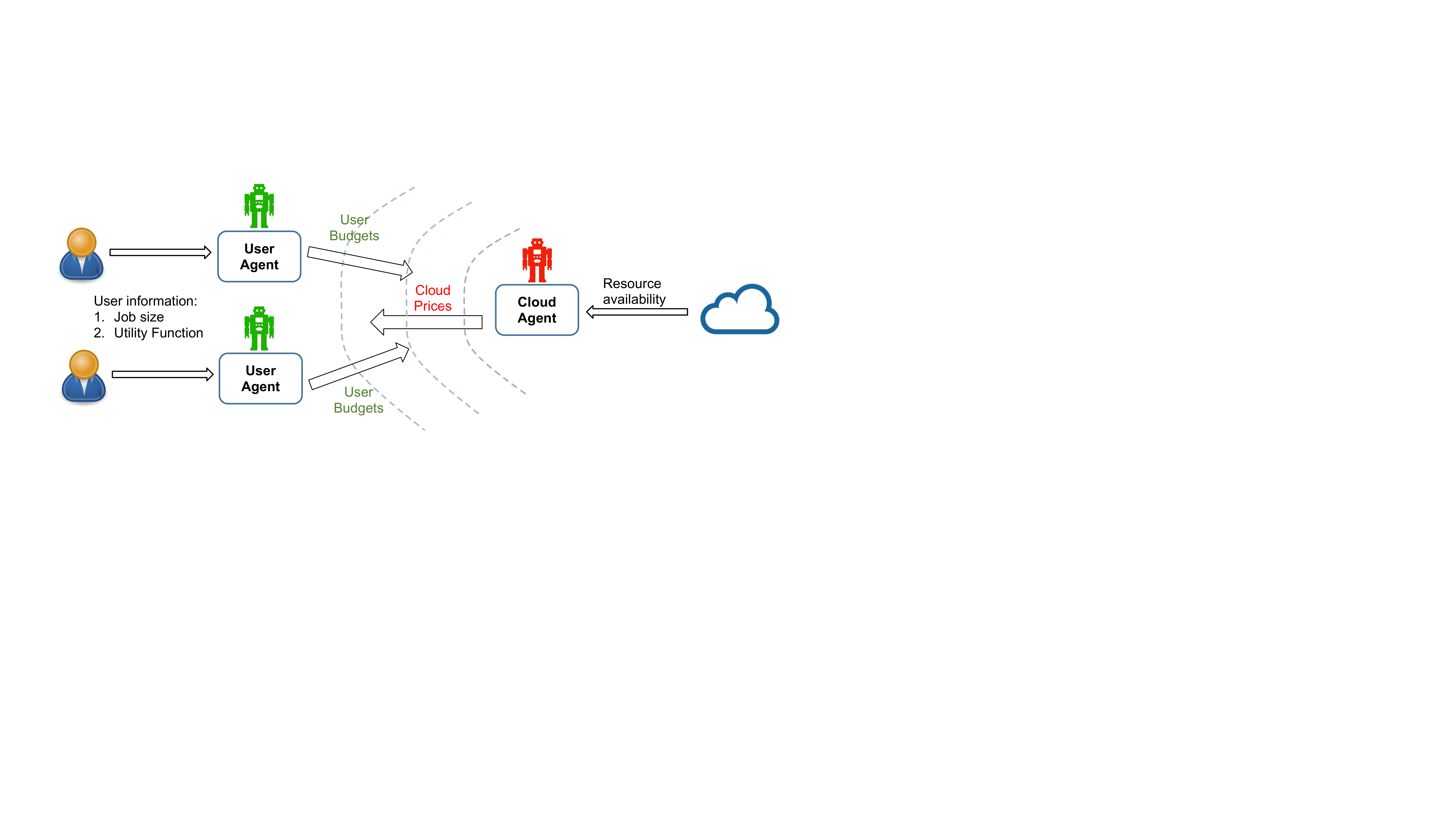}
 \caption{User agents (on behalf of each user) and the cloud agent (on behalf of the cloud) interact in a purely decentralized fashion to allocate resources and decide prices.
}
 \label{fig:decentralized_resource_allocation}
 \end{figure*}

{\bf Proposed decentralized resource allocation and pricing}: 
We envision user agents working on behalf of each user and a cloud agent working on behalf of the cloud. For example, these agents can be robots running a fixed algorithm.
Locally, on the user's end, we let each user choose her utility (or willingness to pay) from a list of options that best suits her needs.
Using this information, the user agent computes the optimal budget response to the prices published by the user and transfers only the budget signals to the cloud keeping the utility information encrypted locally.
Besides, an attractive feature of such a scheme is that the users can update their utility function at their disposal and the user agent will use the most recent choice.
An illustration is provided in Fig. \ref{fig:decentralized_resource_allocation}.
Our algorithm can be interpreted as a protocol between the users and the cloud provider for optimal resource allocation and fair pricing (further explained in Sec. \ref{sec:pricing_mechanism}).

\textbf{Incentivizing truthful reporting}: 
~Such market-based mechanisms need to be incentive-compatible, that is, we need to ensure that the pricing scheme is such that the users cannot gain by faking their utility or willingness to pay.
%
We note along these lines that in our proposed pricing scheme the users are not necessarily charged by the full amount of their willingness to pay. Instead, the amount charged depends on the net demand and the willingness to pay of all the users with the additional feature that no user under any circumstances is charged more than her willingness to pay. Hence, if a user has a high-willingness-to-pay, she should still report it truthfully since the actual charges might be significantly smaller if the demand is less.
Further, if we assume that no user is large enough to single-handedly affect the prices, which is a general assumption in the market-based settings, then it will be evident that the users should respond according to their true utility functions to maximize their personal gain (see USER(i) problem in Sec.~\ref{sec:pricing_mechanism}).

To summarize our key contributions:
\begin{enumerate}
    \item We formulate our scheduling problem as an optimization problem to maximize the net social welfare, i.e. the sum of the utilities of all the users, where  the utility functions incorporate the delay-sensitivity of the users.
    We relax the optimization problem so that is amenable to user-cloud decomposition. (Sec.~\ref{sec: prob_form}.)
    \item We identify the sparsity structure in the solution and use it to show that the relaxed solution is close to the original solution.
    In particular, we provide an upper bound on the gap between the original optimization problem and its relaxation. (Sec.~\ref{sec:analyse}.) 
    \item We interpret the intermediate Lagrange variables in the user-cloud decomposition as pricing and bidding signals.
    Based on this, we propose an algorithm to allocate resources in the subsequent runs of the scheduler and give 
    a gradient-based pricing scheme to track the equilibrium prices. (Sec.~\ref{sec:pricing_mechanism}.)
    \item We demonstrate the viability of our approach through simulations. We give a market simulation where the users' utilities change every day. 
    We compare our allocation scheme with existing schemes (such as first-come-first-serve) 
    that do not account for users' utilities.
    We observe that the sum utility of the system obtained by using our allocation scheme almost doubles the sum utility obtained otherwise. Further, we show that our utility-agnostic approach successfully tracks the variations in users' utilities across days. (Sec.~\ref{sec:sim}.)
\end{enumerate}

We conclude our paper in section~\ref{sec: conclusion} with a few exciting open problems of interest to the research community like machine heterogeneity and job dependencies, and other important considerations for pricing like wholesale discount and user risk-preferences and fault-tolerance.
We give a brief survey of related works in Appendix~\ref{sec: related work}.
To maintain continuity in reading, we defer the
additional remarks and proofs to
Appendices~\ref{sec: remarks} and \ref{sec: proofs}, respectively.

\section{Problem Formulation}
\label{sec: prob_form}

Serverless  (also called Function-as-a-service) as a cloud computing framework is ideal for ``simple'' jobs 
where each \emph{user} submits a \emph{function} to be executed on serverless workers. 
Each user can request for a job comprising of any number of executions of her function at any time, which could be triggered due to external events.
For example, the users can provide the conditions under which they require the execution of a certain number of instances of their function.
The users provide these \emph{trigger event} details along with their function submissions.
Our goal here is to design an efficient real-time job scheduler to allocate resources to the jobs that have been triggered across multiple users and a corresponding pricing scheme for the cloud provider.

We envision a job scheduler that operates periodically and schedules the jobs that are currently in its queue.
This queue consists of all the jobs that have been triggered and are ready to be executed but have not been scheduled yet.
Thus, it consists of previously unscheduled jobs (complete or partial) plus any new jobs that arrived since the previous run of the scheduler.

Our model assumes that the user derives utility only when her job is completed, that is, when all the functions comprising her job are executed where each function execution requires one serverless worker\footnote{Without loss of generality, we make a simplifying assumption that all users' functions are identical in terms of computation costs. Thus, the pricing scheme in a given tier charges users only on the basis of their job sizes.}.
Let the respective delay-sensitivity for user $i$ be captured by a utility function $U_i: [0, \infty) \to \bbR$. 
That is, user $i$ obtains utility $U_i(\tau)$ if her job is completed at time instant $\tau (> 0)$, where $U_i(\cdot)$ is non-increasing.
Let $J_i$ denote the number of function executions needed to complete the job of user $i$.
We call this the size of user $i$'s job.
For example, it could be a single function instance in which case $J_i = 1$ or a batch job of size $J_i > 1$.

We think of the scheduler as allocating resources to the users in different \emph{service tiers} based on their execution times---some jobs will be scheduled for immediate execution, whereas others will be scheduled for execution at later times in the future.
The jobs that are not scheduled will remain in the queue to be scheduled by the scheduler at later operations.
Note that the jobs that have been scheduled for execution at a future point are removed from the queue.
Let $T$ be the maximum number of service tiers offered by the cloud provider, and let the $t$-th tier be characterized by the end time of this tier given by $\tau_t$.
For example, if $T = 5$, then the five tiers can correspond to job completion under 1 second, 10 seconds, 10 minutes, 1 hour and 10 hours, respectively. 
Notice that the end times of the different tiers in our model need not be evenly spaced.
This is a useful feature because it allows our scheduler to plan over longer time horizons with a limited number of tiers.
The pricing for these different tiers of service should ideally decrease as the job completion time increases.

We refer to the intervening time between the adjacent tier end times as \emph{service intervals}.
Continuing the previous example, service interval 1 starts at 0 seconds and ends at 1 second, service interval 2 starts at 1 second and ends at 10 seconds, etc.
Note that the functions in a single user's job can be served across several service intervals, but has a unique tier when all these functions are completed.
Further, let $M_t$ be the constraint on the number of machines available for the scheduler to allocate resources in service interval $t$. 
For example, if 500 machines are available every 100ms and tier 1 is one second, then $M_1= 500*(1/0.1)= 5000$. 
Similarly, if tier 2 is 10 seconds, then the number of machines available for tier 2 jobs are $500*(10/0.1) - 5000 = 45,000$.
However, we would like to retain some portion of this for the subsequent operations of our scheduler.
This will allow the scheduler to allocate resources for urgent jobs arriving in the future.
Say we decide to utilize at most $20 \%$ of these machines, then $M_2 = 45,000*(0.2) = 9000$.

\begin{remark}
The values of $M_t$'s can depend on several constraints that are subjective to the cloud provider.
In general, $M_t$ can be calculated on the basis of 1) machine occupancy by pre-scheduled jobs, 2)
 the forecast of the number of urgent jobs in subsequent time, which can be modeled on the basis of historical data available at the cloud provider, and 3) a margin to account for system uncertainties such a power shutdown, machine failures, etc. 
\end{remark}

{\bf System Problem}:
At each implementation of the scheduler, we assume that the scheduler has access to the unscheduled jobs ready to be executed (say job of size $J_i$ for user $i$), their utility functions adjusted for past delay (for example, if a job has been waiting in queue for time $\tau_0$, then we take its utility function to be $U_i(\tau + \tau_0)$), and the machine availability $M_t$ for each service interval $t$.
The scheduler then outputs a feasible allocation of serverless workers across different service intervals.
Let $x_{i,t}$ denote the number of functions 
executed
for agent $i \in [N] := \{1, 2, \dots, N\}$ at service interval $t \in [T] := \{1, 2, \dots, T\}$.
Let $\x$ be the $(N\times T)$-matrix with entries $x_{i,t}$, $i \in [N], t \in [T]$.

Let $U_{i,t} := U_i(\tau_t)$ denote user $i$'s utility if her job is completed in tier $t$, where $\tau_t$ is the end time of tier $t$.
Then 
\begin{equation}
\label{eq: end_time_def}
	T_i := \min \{t \in [T] : \sum_{s = 1}^t x_{i,s} \geq J_i \},
\end{equation}
 is the time it takes to complete user $i$'s $J_i$ functions, awarding her a utility of $U_{i,T_i}$. 
(If a user's job is not completed we let $T_i = T + 1$ and assign her zero utility.)
Also, at any service interval $t\in [T]$, since the number of function executions cannot exceed the cloud service provider's capacity, we have  $\sum_{i = 1}^N x_{i,t} \leq M_t$.
To this end, we formulate the system problem SYS that maximizes the sum utility of the system as follows (see Fig. \ref{fig:sys_prob} for an illustration):
{
\allowdisplaybreaks
\begin{align}
\text{\bf\underline{SYS}} \nn \\
\maxi_{x_{i,t}\geq 0} ~~~~~& \sum_{i=1}^n U_{i,T_i}\nn \\
\text{subject to} ~~~~~& \sum_{i = 1}^N x_{i,t} \leq M_t, \forall~ t \in [T], \text{ and }\label{eq: sys_capacity}\\
 & T_i = \min \{t \in [T] : \sum_{s = 1}^t x_{i,s} \geq J_i \}, ~\forall~ i \in [N]. \label{eq: sys_T_def}
\end{align}
}

In general, SYS can have multiple solutions, but there always exists a solution that 
assumes a special form. 
Specifically, it corresponds to a \emph{non-preemptive scheduling}, where the resources are allocated so that none of the users' jobs are interrupted in the middle of the execution and made to wait till its execution resumes at a later stage after serving other users.
We first outline a rough sketch of the underlying intuition and then establish this result in Lemma~\ref{lemma:ordering_sys}.
(See \cite{mcnaughton1959scheduling} for a similar result in the context of scheduling jobs on a single machine.)

Since the utilities of users depend only on their job completion times, it is suboptimal to leave a user's job unfinished by allocating partial resources (which provides no utility gain). 
For example, say user $i$ has been allocated a total of $\sum_{s=1}^t x_{i,s} (<J_i)$ machines till the end of service interval $t$, for some $1 \leq t < T$, and then interrupted, i.e. $x_{i,(t+1)} = 0$.
Further, let there be
a user $j (\neq i)$ that is allocated resources in the service interval $t+1$.
If user $i$'s job completes before user $j$, then we can swap the function executions of user $i$ at times $> t+1$ with the foremost function executions of user $j$ without reducing the sum utility, since such a swap will only possibly reduce the completion time of user $i$ without affecting that of user $j$.
On the other hand, if user $j$'s job completes before user $i$, then we can swap the function executions of user $i$ at times $\leq t$ with the hindmost function executions of user $j$ without reducing the sum utility, since such a swap will only reduce the completion time of user $j$ without affecting that of user $i$.
We thus have the following lemma that we prove formally in Appendix~\ref{sec: proofs} (see also Remark~\ref{rem: greedy_sparse}).

\begin{lemma}[Non-preemptive scheduling]
\label{lemma:ordering_sys}
There exists an optimal solution to SYS that allocates resources to users in an uninterrupted fashion, that is, if user $i$ gets allocated some resources, the system would allocate $J_i$ resources to user $i$ (possibly across multiple service intervals) instead of halting it and attending other users.
\end{lemma}

Let $u_{i,t} := U_{i,t} - U_{i,{t+1}},~\forall i,t$, where we assume that $U_{i,{T+1}} := 0, \forall i,$
consistent with the fact that we assign zero utility if the jobs are not completed.
Thus, $u_{i,t} \geq 0, \forall i,t,$ as $U_i(\cdot)$'s are monotonically decreasing.
Condition~\eqref{eq: sys_T_def} in problem SYS is not favorable from an optimization framework point of view.
Hence, we introduce an indicator variable $y_{i,t}, \forall i,t$, that we use as a proxy to indicate whether user $i$'s job is completed on or before time $t$. 
Let $\y := (y_{i,t})_{i \in [N], t \in [T]}$.
With these variables, the scheduler optimization can be formulated as follows:
{
\allowdisplaybreaks
\begin{align}
\text{\bf {\underline {SYS-ILP}}}\nn \\
 \maxi_{x_{i,t}\geq 0, y_{i,t} \in \{0,1\} }
 ~~~~~&\sum_{t=1}^T\sum_{i=1}^n u_{i,t}y_{i,t} \nn \\
 \text{subject to} ~~~~~& y_{i,t} \leq \frac{\sum_{s = 1}^t x_{i,s}}{J_i}, \forall i \in [N], t \in [T], \text{ and }\label{sys_ilp_y}\\
& \sum_{i = 1}^N x_{i,t} \leq M_t, ~\forall ~t \in [T]. \label{sys_ilp_x}
\end{align}
}

If $\x$ is a feasible solution for SYS, then defining $y_{i,t}, \forall i,t$, to be equal to $1$ if user $i$'s job is completed by time $t$ and equal to zero otherwise, we observe that $\x, \y$ is a feasible solution to SYS-ILP with the same objective value as that of SYS with $\x$.
On the other hand, if $\x, \y$ is a feasible solution for SYS-ILP, then defining $T_i = \min\{t\in [T] : y_{i,t}>0\}$ for all $i\in[N]$, we get that it forms a feasible solution for SYS with the same objective value.
This gives us an equivalence between the problems SYS and SYS-ILP.
Problem SYS-ILP can be solved using Mixed Integer Linear Programming (MILP) methods available in computing frameworks such as MATLAB \cite{scip_milp} for sufficiently small $N$ and $T$.

{\bf Relaxing SYS-ILP}:
Although SYS/SYS-ILP are NP-hard in general (see Remark~\ref{rem: NP_hardness}), they can be solved approximately by relaxing the integer constraints.
%
%
%
Let us 
replace the constraint $y_{i,t} \in \{0,1\}$ by $0 \leq y_{i,t} \leq 1$. 
%
Since $u_{i,t} \geq 0$, the optimal $y_{i,t}^*$ for the relaxed problem is attained when the constraint \eqref{sys_ilp_y} is satisfied with equality for all $i, t$ as long as $y_{i,t} \leq 1$.
Hence, we can substitute $y_{i,t} = \textstyle\sum_{s = 1}^t x_{i,s}/J_i$ in the objective function and introduce an additional constraint, $\textstyle\sum_{t = 1}^T x_{i,t} \leq J_i$, for all i, to ensure that $y_{i,t} \leq 1$.
Fractional $y_{i,t}$ represents the fraction of functions completed for user $i$ at time $t$, 
and the system is awarded a utility per function depending on the tiers in which these functions are executed unlike depending on the completion time of all functions considered earlier.
Letting 
\[
	F_{i,t} := \frac{U_{i,t}}{J_i}
\] 
denote the utility per unit function, for all $i$ and $t$, and making appropriate substitutions in SYS-ILP, we get
\begin{align}
\text{\bf \underline{SYS-LP}}\nn\\
\maxi_{x_{i,t}\geq 0} ~~~~~&\sum_{i=1}^N \sum_{t=1}^T x_{i,t}F_{i,t}\nn \\
\text{subject to} ~~~~~&\sum_{t=1}^T x_{i,t} \leq J_i, ~\forall ~i \in [N], \label{tot_job_cons}\\
&\sum_{i = 1}^N x_{i,t} \leq M_t, ~\forall~ t \in [T]. \label{tot_resource_cons}
\end{align}
The nice LP form of SYS-LP would later allow us to devise pricing mechanisms for resource allocation, similar to the network resource allocation schemes in \cite{kelly1997charging, kelly1998rate}.  
In particular, the dual variable $\mu_t$ corresponding to the constraint \eqref{tot_resource_cons} plays the role of an auxiliary price (or a shadow price) for the $t$-th service tier (more on this in Sec. \ref{sec:pricing_mechanism}).

Let $V^*$ be the optimum sum utility of the system in SYS/SYS-ILP.
Let $V^R$ denote the optimum value for the relaxed problem SYS-LP (which is obtained by relaxing the integer constraints on $y_{i,t}$).
It is clear that $V^R \geq V^*$, since $V^R$ is obtained by relaxing the integer constraints.
Let $\x^R$ be the corresponding resource allocation matrix that achieves the optimum value $V^R$ in SYS-LP.
Let $\y^R$ be the corresponding matrix given by
\begin{equation}
\label{eq: y_R_def}
	y^R_{it} := \sum_{s = 1}^t \frac{x_{i,s}}{J_i}.
\end{equation}
We define the matrix $\hat \y \in \bbR^{N \times T}$ with entries 
\begin{equation}\label{y_hat}
	\hat y_{i,t} := \begin{cases}
		0, \text{ if } y^R_{i,t} < 1,\\
		1, \text{ if } y^R_{i,t} = 1,
	\end{cases} \forall i, t.
\end{equation}
We can verify that $\x^R, \hat \y$ is a feasible solution for SYS-ILP. Let $\hat V$ denote the value of the objective in SYS-ILP at $\x^R, \hat \y$.
Then, since $\hat V$ is one possible solution to SYS-ILP (whose optimal solution is $V^*$), we have 
\begin{equation}\label{relation_objs}
V^R \geq V^* \geq \hat V.
\end{equation}
Later, we show that the inequality gap $(V^R - \hat V)$ is small. Note that this implies that the gaps $(V^* - V^R)$ and $(V^* - \hat V)$ are small. 
To that end, we first analyze the problem SYS-LP in more detail.
(See Remark~\ref{rem: compare_wang} for an alternative way to relax the problem SYS as proposed by \cite{wang2012datacenter} and how it compares with our approach.)

\section{Analyzing SYS-LP}
\label{sec:analyse}
Consider the Lagrangian corresponding to the optimization problem SYS-LP,
{\small
\begin{align}
L(\x, \bl, \bmu) &= \sum_{i=1}^N \sum_{t=1}^T x_{i,t}F_{i,t} + \lambda_i\left(J_i - \sum_{i=1}^N x_{i,t}\right) \\ 
&+ \mu_t\left(M_t - \sum_{t=1}^T x_{i,t}\right)\nn \\
&= \sum_{i=1}^N \sum_{t=1}^T \left(F_{i,t} - \mu_t - \lambda_i\right)x_{i,t} + \sum_{t=1}^T\mu_tM_t + \sum_{i=1}^N\lambda_i J_i, \label{lag:system}
\end{align}
}\\
where $\x = [x_{i,t}], i\in[N], t \in [T],$ is the primal variable matrix and $\bl = [\lambda_i], i\in [N]$, and $\bmu = [\mu_t], t\in [T]$, are the dual variable vectors corresponding to the  constraints in \eqref{tot_job_cons} and \eqref{tot_resource_cons}, respectively. 
Here, $\mu$ is the auxiliary price vector for the $T$ service tiers that comes out of the resource allocation problem. 
Thus,
{
\begin{align*}
\frac{\partial L}{\partial x_{i,t}} &=  F_{i,t} - \mu_t - \lambda_i, ~~\forall i, t;\\
\frac{\partial L}{\partial \lambda_i} &= J_i - \sum_{i=1}^N x_{i,t},~~~~~~~ \forall i;\\
\frac{\partial L}{\partial \mu_t} &= M - \sum_{t=1}^T x_{i,t}, ~~~~~\forall t.
\end{align*}
}\\
Hence, the Karush-Kuhn-Tucker (KKT) conditions \cite{boyd} imply that
\begin{align}
\label{sys_kkt1} 
\mu_t + \lambda_i 
&\begin{cases}
= F_{i,t}, & \text{ if } x_{i,t} > 0,\\
\geq F_{i,t}, & \text{ if } x_{i,t} = 0,
\end{cases}
\quad \quad \forall i, t,\\
\label{sys_kkt2} 
\sum_{i=1}^{N} x_{i,t} 
&\begin{cases}
= M_t, & \text{ if } \mu_{t} > 0,\\
\leq M_t, & \text{ if } \mu_{t} = 0,
\end{cases}
\quad \quad ~\forall t,\\
\label{sys_kkt3}
\sum_{t=1}^{T} x_{i,t} 
&\begin{cases}
= J_i, & \text{ if } \lambda_{i} > 0,\\
\leq J_i, & \text{ if } \lambda_{i} = 0,
\end{cases}
\quad \quad ~~~\forall i.
\end{align}

We will now show that there exists an optimal solution to the SYS-LP problem with a special structure.
 Note that when $x_{i,t} > 0$, we have $\mu_t + \lambda_i = F_{i,t}$ from the KKT condition in Eq. \eqref{sys_kkt1}. 
We know that every time a user $i$ gets non-zero resources allocated in multiple tiers, say $t_1$ and $t_2$, we have $\mu_{t_1} + \lambda_i = F_{i,{t_1}}$ and $\mu_{t_2} + \lambda_i = F_{i,{t_2}}$. 
This implies that 
\begin{equation}
\label{eq: mu_relation}
	\mu_{t_1} - \mu_{t_2} = F_{i,{t_1}} - F_{i,{t_2}}.
\end{equation}
Hence, every time any user gets partial resources (that is, less than the total job size) allocated in one tier, we get one relation of the form \eqref{eq: mu_relation}.
If there are more than $T-1$ instances of such partial resource allocation, then we can eliminate the variables $\mu_t, t \in [T]$, to get a non-trivial relation amongst the variables $F_{i,t}$'s.
In a generic case, where the $F_{i,t}$'s are any real numbers, it is unlikely that such a non-trivial relation amongst the variables $F_{i,t}$'s is satisfied.
Expanding on this idea, we prove the following lemma whose proof is included in Appendix~\ref{sec: proofs}. (See \cite{wang2011cost} for a related result in the context of network resource allocation. See also Remark~\ref{rem: optimal_soln_sparsity}.)

\begin{lemma}
\label{lem: sparsity}
	There exists an optimal solution $x_{i,t}$ to SYS-LP such that $|\{ (i,t) : 0 < x_{i,t} < J_i \}| \leq T$.
\end{lemma}

Recall that as a result of Lemma \ref{lemma:ordering_sys}, we know that the optimal resource allocation in SYS is $(N+T)$-sparse. 
Sparsity in both SYS and SYS-LP solutions suggests that the resource allocation in the two cases is similar. 
In Theorem \ref{thm:gap_bound}, we formally bound the gap, $V^R - \hat V$.

\begin{theorem}\label{thm:gap_bound}
Let $V^*$ and $V^R$ be the optimal objectives of the problem in SYS and SYS-LP, respectively. 
Then, we have
\begin{equation}
\label{eq: bound_gap}
\hat V \geq \left(1 - \frac{T(\max_i J_i)}{\min_t M_t}\right) V^R.
\end{equation}  
\end{theorem}

Thus, the relaxation from SYS-ILP to SYS-LP does not sacrifice much in terms of optimality. This is especially true in practical settings where the number of tiers, $T$, is small (e.g., three tiers corresponding to slow, medium or fast service) and the number of users (vis-à-vis the number of machines, $M_t$) is large.
We observe something similar in Fig. \ref{fig:error_v}, where we plot the mean and standard deviation of the percentage error between $V^R$ and $\hat V$, that is, $100*(V^R - \hat V)/V^R$ over 20 independent trials for $T=5$ tiers and $N=100$ users, where utilities and job sizes are independently and identically distributed (i.i.d.) across users. $J_i, ~\forall ~i\in [N]$, are uniformly chosen between 10 and 90. We have taken $M_t(= M)$ constant across tiers.
Recall that the plot is an upper bound on the percentage error between $V^R$ and $V^*$.
Further, under an additional constraint on allocation, we derive a tighter bound in Remark~\ref{rem: tighter_bound} which is independent of the number of tiers, T.
\begin{figure}[t]
    \centering
    \includegraphics[scale=0.8]{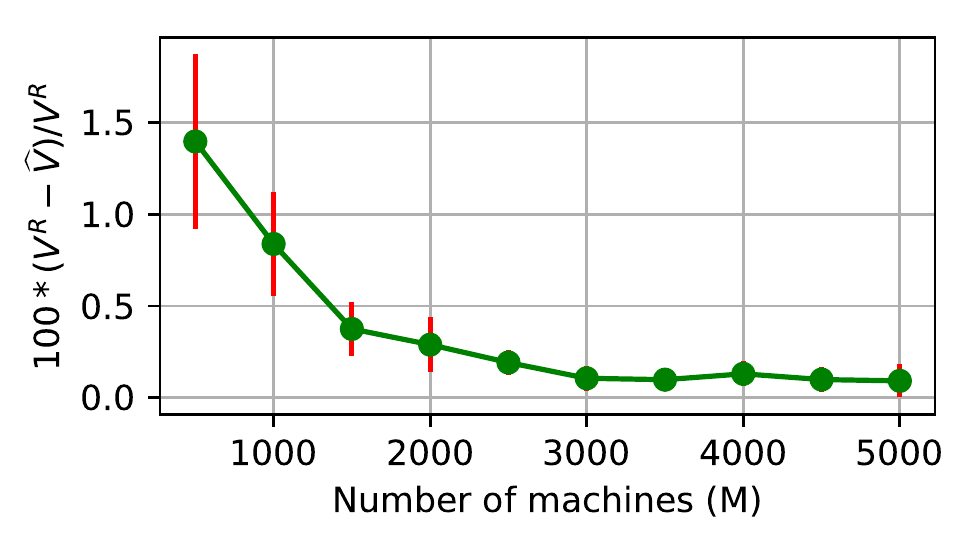}
    \caption{Percentage error between $V^R$ (the utility obtained with SYS-LP problem) and $\hat V$ (the utility obtained by applying integer constraints on SYS-LP solution).
    }
    \label{fig:error_v}
\end{figure}


{\bf Premium services are charged more}: We now give a toy example to illustrate some of the features of our pricing scheme. 
Consider a system with $N=3$ users having utility functions as shown in Fig. \ref{fig:utility_examples}. 
Let the total job size for each user be $10$. 
Let the cloud provider allocate resources in $T=3$ tiers corresponding to job completion under $0.1, 10$ and $1000$ seconds. 
Also, let $M_t = 10$ for $t=1,2,3$. Thus, according to the utilities shown in Fig. \ref{fig:utility_examples}, the $3\times 3$ utility matrix $\mathbf{U} = [U_i(t)|~ i\in \{1,2,3\}, t\in \{1,2,3\}]$ is given by
$$
\mathbf{U} = 
\begin{bmatrix}
3 & 0 & 0\\
4 & 2.5 & 1\\
2 & 2 & 2
\end{bmatrix}.
$$
For both SYS and SYS-LP, the following resource allocation is optimal: \begin{equation}\label{x_it_fig1}
	 x_{i,t} := \begin{cases}
		10, \text{ if } i=t,\\
		0, \text{ if } i\neq t,
	\end{cases} \forall~ i, t \in \{1,2,3\}.
\end{equation}
The above allocation makes intuitive sense because it maximizes the sum utility of the system by maximizing the individual utility of the users.
Any vector $\bmu = (\mu_1, \mu_2, \mu_3)$ such that $0.15 \leq \mu_1 \leq 0.3$, $0 \leq \mu_2 \leq 0.25$, $0 \leq \mu_3 \leq 0.2$, $\mu_1 - \mu_2 \geq 0.15$ and $\mu_2 \geq \mu_3$ serves as an auxiliary pricing vector for the three tiers corresponding to the above optimal resource allocation.
(Observe that if we take $\lambda_1 = 0.3 - \mu_1$, $\lambda_2 = 0.25 - \mu_2$, and $\lambda_3 = 0.2 - \mu_3$, then the KKT conditions \eqref{sys_kkt1}, \eqref{sys_kkt2}, and \eqref{sys_kkt3} are satisfied.)
An example of such an auxiliary pricing vector for the three service tiers is given by 
$\bmu = [0.259, 0.083, 0.048]$ (obtained through the \texttt{CVX} optimizer \cite{cvx,gb08}).
Note that in our scheme, the cloud provider sets prices for different service tiers proportional to $\bmu$. Thus, user 1 is charged 
more for getting her job completed in tier 1 as opposed to user 3 who is charged less for being flexible. Note that this is also better than a myopic greedy allocation where user 2 would have been selected first based on her high utility value in the first tier.
Our pricing scheme thus takes into account the users' delay-sensitivity and charges each user accordingly.
The conditions on $\bmu$ show that a flat pricing scheme cannot achieve this result in our toy example setting.

\section{Pricing Mechanisms to Learn User Utilities}\label{sec:pricing_mechanism}

In the previous section, we assumed that the utilities of the users are known to the cloud service provider. This is not true in general.
However, through dynamic pricing and users' responses to this pricing, the cloud service provider can converge to the optimal scheduling and pricing.
Such pricing mechanisms have been studied in economics for two-sided (supply and demand) markets and are called Walrasian auctions \cite{walrasian, microeconomic_mascolell}. In the networking 
literature, Kelly has used similar mechanisms to allocate bandwidth optimally over a network \cite{kelly1997charging,kelly1998rate}. 

In this section, we devise dynamic pricing schemes that would help the cloud provider to maximize the utility of the system. 
To the best of our knowledge, \emph{this is the first work to study dynamic pricing mechanisms for scheduling problems}. This is a promising research direction with potentially high impact on the economics and resource allocation of next-generation computing systems.
Next, we derive an algorithm that sets pricing for cloud resources without assuming that the users' utilities are known.

We first decompose the system problem from the $i$-th user's perspective and from the cloud provider's perspective. 
To introduce the user feedback, we define a new variable $m_{i,t}$, which can be thought of as the budget of user $i$ for service interval $t$. 
Let $q_t$ be the price set by the cloud service provider at service interval $t$.
Then the amount of resources allocated to the $i$-th user in service interval $t$ is given by $x_{i,t} = m_{i,t}/q_t$. 
With this interpretation in mind, 
we formulate the user problem for the $i$-th user by taking the terms from the Lagrangian of the system formulation (Eq. \eqref{lag:system}) that are relevant to the $i$-th user along with the constraint on the number of functions in her job (Eq. \eqref{tot_job_cons}).
\begin{align}
\text{\bf \underline{USER(i)}}\nn\\ 
\maxi_{m_{i,t}\geq 0} ~~~~~&\sum_{t=1}^T \frac{m_{i,t}}{q_t}\left(F_{i,t} - q_t\right)  \nn \\
\text{subject to} ~~~~~&\sum_{t=1}^T \frac{m_{i,t}}{q_t} \leq J_i \label{user_prob_cons}.
\end{align}
The $i$-th user thus allocates a larger budget $m_{i,t}$ at time $t$ if the user utility per function at time $t$, $F_{i,t}$, is sufficiently larger than the price $q_t$ set by the cloud service provider. Note that, if $q_t = 0$, then $m_{i,t} = 0$. Because, otherwise, we would have $m_{i,t}/q_t = \infty$, and this would violate constraint \eqref{user_prob_cons}. Given a price vector $\q = [q_t], t\in [T],$ the $i$-th user solves the USER(i) problem to obtain the budget vector $m_{i,t}, t\in [T]$. The cloud provider then receives budgets from all users $m_{i,t}, t\in [T], i\in [N],$ and solves the CLOUD problem defined as follows: 
\begin{align}
\text{\bf \underline{CLOUD}}\nn\\  
\maxi_{x_{i,t}\geq 0} ~~~~~&\sum_{t=1}^T\sum_{i=1}^N m_{i,t}\log x_{i,t}  \nn \\
\text{subject to}~~~~~& \sum_{i=1}^N x_{i,t} \leq M, ~\forall~ t \in [T] \label{cloud_prob_cons}.
\end{align}

\begin{theorem}
\label{thm:decomposition}
There exist equilibrium matrices $\x = (x_{i,t}, i \in [N], ~t \in [T])$ and $\m = (m_{i,t}, i\in [N], t\in [T])$, and an equilibrium price vector $\q = (q_t, t\in [T])$ such that
\begin{enumerate}[(i)]
	\item $\m_i = (m_{i,t}, t\in [T])$ solves USER(i), $\forall~ i\in [N],$
	\item $\x = (x_{i,t}, i\in [N], t\in [T])$ solves the CLOUD problem,
	\item $m_{i,t} = x_{i,t} q_t,$ for all $i\in [N], t\in [T]$,
	\item for any $t$, if $\sum_i x_{i,t} < M_t$, then $q_t = 0$.
\end{enumerate}
Further, if any matrix $\x = (x_{i,t}, i\in [N], t\in [T])$ that is at equilibrium, i.e. has a corresponding matrix $\m$ and a vector $\q$ that together satisfy (i), (ii), (iii), and (iv), then $\x$ solves the system problem SYS-LP.
\end{theorem}

Condition (i) says that the equilibrium budgets $(m_{i,t}, t \in [T])$ form an optimum response by the user $i$ to the equilibrium rates $(q_t, t \in [T])$.
Condition (ii) says that the equilibrium allocation matrix $\x$ is the solution to the CLOUD problem with respect to the equilibrium budgets $\m$.
Further, the allocations, prices and the budgets are consistent under equilibrium, i.e. $m_{i,t} = x_{i,t} q_t$ as per our interpretation of budgets (from condition (iii)).
And finally, (iv) says that the price $q_t = 0$ if that corresponding service interval is not full, i.e. $\sum_i x_{i,t} < M_t$.
Besides, in the proof of Theorem \ref{thm:decomposition}, we show that if $\bmu$ is the dual variable corresponding to an optimal solution $\x$ of SYS-LP, then $\q = \bmu$ form an equilibrium price vector.

{\bf Price tracking using gradient descent}: In many cases, it is easier to work with the dual problem because there are less variables involved. For example, for the CLOUD problem, the primal has $NT$ variables, which can be significantly large in comparison to only $T$ variables in the dual problem (coming from $T$ constraints in the primal). Moreover, working on the dual problem allows the cloud provider to work directly on the price vector $\q$, which is what it shows to the users.
We can derive the following dual problem for the CLOUD problem using the Lagrangian $L(\x,\q)$ in Eq.~\eqref{lag:cloud}. 
\begin{align}\label{cloud_dual}
\maxi_{q_t\geq 0}~~~ & \sum_{t=1}^T\log q_t \left(\sum_{i=1}^N m_{i,t}\right) - \left(\sum_{t=1}^T M_t q_t\right).
\end{align}
The gradient of the above objective w.r.t. $q_t$ is given by $\left(\frac{\sum_{i=1}^N m_{i,t}}{q_t} - M_t\right)~\forall~t\in[T].$ We use gradient descent to solve Eq. \eqref{cloud_dual} and make the pricing scheme $\q$ track users' budgets $\m$.

Next, based on our results in this section, we derive an algorithm where the cloud provider updates the prices of its resources. Here, we assume that the user utilities are not known to the cloud provider (which is generally the case) but they provide their budget information (represented as the $\m$ matrix in the USER/CLOUD problems) on a semi-regular basis which is determined based on the current price vector (represented as $\q$ in the USER/CLOUD problems) shown by the cloud service provider (see Alg. \ref{algo:price_tracking}).

\begin{algorithm}[ht]
\SetAlgoLined
\SetKwInOut{Input}{Input}
\Input{Total number of machines $M$, step-size $\kappa$, error tolerance $\epsilon$, gradient steps $G$, job sizes $J_i$ and utilities $U_i(\cdot)$ that are known only to user $i$ for $i\in[N]$}
\textbf{Initialization}: Cloud service provider shows some initial prices $\q = \q_0\in\R^T$ (say, the all ones vector)\\
  \While{$\|\q - \q_{prev}\| \geq \epsilon\|\q_{prev}\|$}{
$\q_{prev} = \q$\\
Users, for all, $i\in [N]$, solve for budget vector $\m_i\in \R^T$ using the price vector $\q$ in the USER(i) problem\\
The cloud service provider receives the budget matrix $\m$ and does the following\\
step = 0\\
\While{step $\leq G$}
{$q_t = \max\left[q_t + \kappa\left(\frac{\sum_{i=1}^N m_{i,t}}{q_t(\tau)} - M_t\right),0\right]~\forall~t\in [T]$\\
$step = step + 1$
}
}
$\q^* = \q, \m^* = \m$\\
$z_{i,t} = m_{i,t}^*/q_t^*$ for all $i\in[N], t\in [T]$\\
Use Algorithm \ref{algo:projection} to obtain $x_{i,t}^*$, the projection of $z_{i,t}$ to the constraint sets $\sum_i z_{i,t} \leq M_t, ~\forall~t$ and $z_{i,t} \geq 0~\forall~i,t,k$. \\
\KwResult{Solution to SYS, that is, the optimal pricing vector $\q^*$ and the optimal resource allocation vector $\x^*$}
 \caption{Finding the optimal resource allocation and pricing without utility information at the cloud end}
 \label{algo:price_tracking}
\end{algorithm}

\begin{algorithm}[ht]
\SetAlgoLined
\SetKwInOut{Input}{Input}
\Input{$\z \in\R^{N\times T}$, Constraint sets $C_1, C_2 \subset \R^{N\times T}$, where $C_1 := \{z_{i,t} |~\sum_i z_{i,t} \leq M_t, ~\forall~t\}$, and $C_2 := \{z_{i,t}|~z_{i,t} \geq 0~\forall~i,t\}$, error tolerance $\epsilon$}
\textbf{Initialization}: Define $\z_1 = \z,~~~ \z_2 = \z$. \\
  \While{$\|\z - \z_{prev}\| \geq \epsilon\|\z_{prev}\|$}{
$\z_{prev} = \z$\\
$(z_1')_{i,t} = z_{i,t} - \max\left[\frac{\sum_{i}z_{i,t} - M_t}{N}, 0\right]~\forall~i,t$.\\
$\z_2' = \z,~~~~ \z_2'[\z_2' < 0] = 0$ \\
$\z = (\z_1' + \z_2')/2$\\
$\z_1 = \z + \z_1 - \z_1'$\\
$\z_2 = \z + \z_2 - \z_2'$
}
\KwResult{$\z$: Projection of the input matrix to sets $C_1$ and $C_2$}
 \caption{Han's algorithm for projection on the intersection of convex sets \cite{han1988successive,phade2017distributed}}
 \label{algo:projection}
\end{algorithm}

An illustration of the process is provided in Fig. \ref{fig:price_tracking}.
Further, for $N=100$ users, $T=5$ hours and i.i.d. users' utilities and job sizes, we plot the results of Algorithm \ref{algo:price_tracking} (with gradient steps $G=40$ and step-size $\kappa=10^{-6}$) in Fig. \ref{fig:price_tracking_perf}. After only $\sim$10 iterations (budget/pricing updates), the tracking-based scheme converges to the optimal solution (Fig. \ref{fig:tracking_util}). Moreover, as shown in Fig. \ref{fig:tracking_perf_price_tier0}, with only slight changes in its pricing scheme, the cloud provider is able to nudge the users to change their budgets to match the optimal solution.

\begin{figure*}[t]
 \centering
 \includegraphics[width=0.75\textwidth]{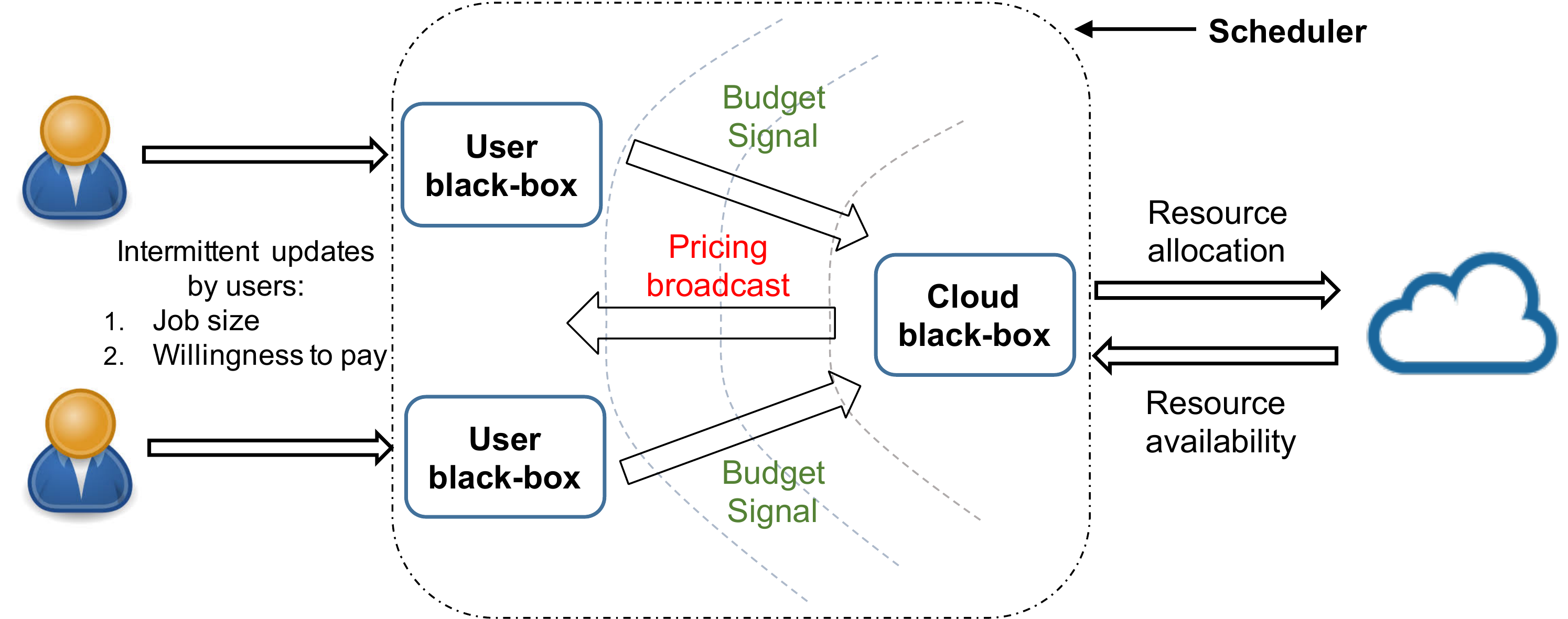}
 \caption{Decomposition into user and cloud problems. Here, we visualize a user black-box for each user $i$ that solves the $\mathrm{USER}(i)$ problem using the prices shown by the cloud black-box, the unscheduled functions in the job of user $i$ and her willingness to pay (intermittently updated). 
The cloud black-box runs Algorithm~\ref{algo:price_tracking} to update the prices
using the budget signals from the user black-boxes and the capacity constraints from the cloud. 
}
 \label{fig:price_tracking}
 \end{figure*}

\begin{figure*}[t]
    \centering
    \begin{subfigure}[t]{0.45\textwidth}
        \centering
        \includegraphics[scale=0.60]{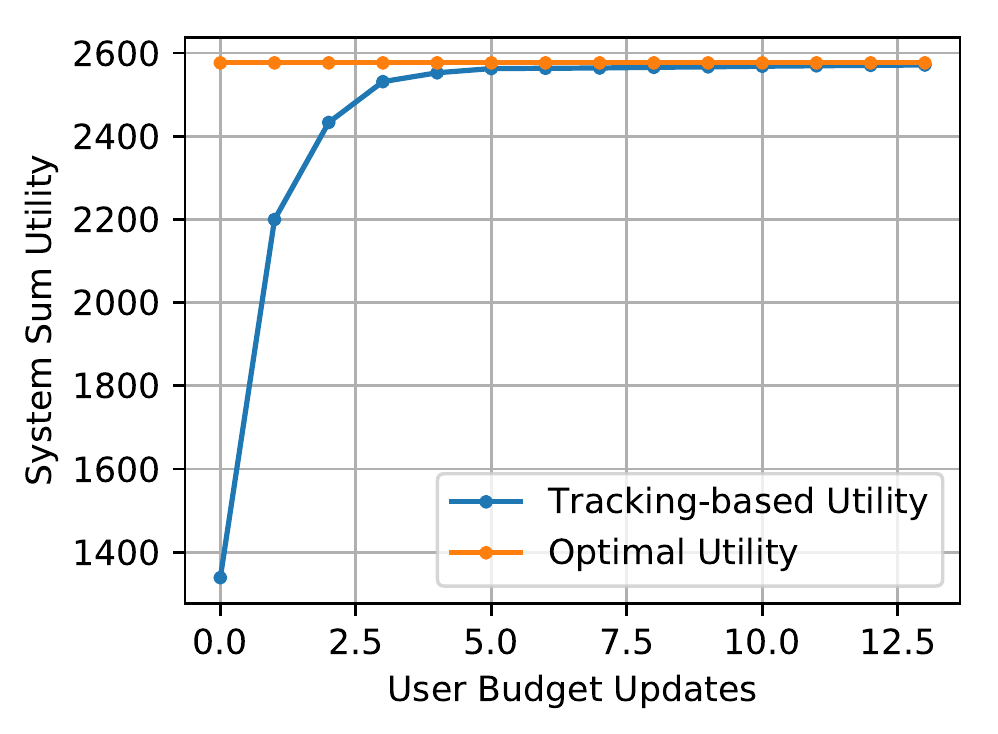}
        \caption{ 
        After less than 10 budget updates, the system utility obtained through Algorithm \ref{algo:price_tracking} approaches the optimal utility.
        }
        \label{fig:tracking_util}
    \end{subfigure}
    ~
    \begin{subfigure}[t]{0.45\textwidth}
        \centering
        \includegraphics[scale=0.60]{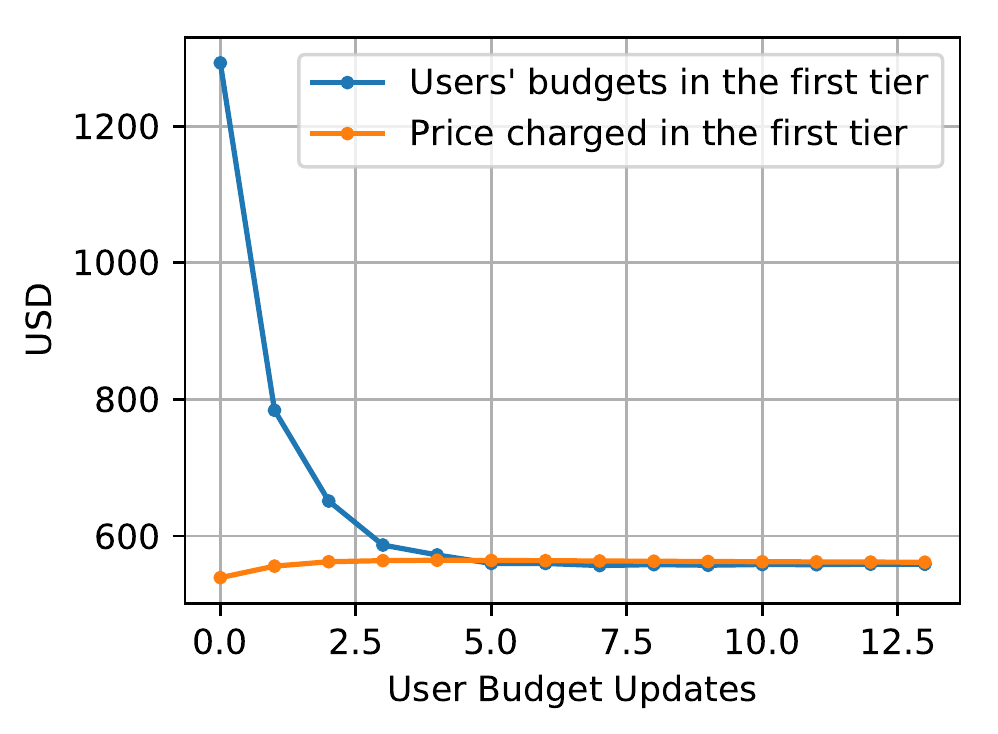}
        \caption{
        After less than 10 budget updates, the prices and users' budgets reach an equilibrium.
        }
        \label{fig:tracking_perf_price_tier0}
    \end{subfigure}
    \caption{Finding the SYS-LP solution through price tracking  (Algorithm \ref{algo:price_tracking}). The algorithm converges after $\sim$10 iterations.}
\label{fig:price_tracking_perf}
\end{figure*}

\section{A Market Simulation}
\label{sec:sim}

In this section, we run an experiment for 60 days with $N = 100$ users and $T=5$ tiers, where each day, users' utilities are changing based on the market trends\footnote{An implementation of the price tracking algorithm (Alg. \ref{algo:price_tracking}) and the market simulation described in this paper is available 
at: https://github.com/vvipgupta/serverless-resource-allocation-and-pricing
}. 
Specifically, for the first 30 days, we synthetically generate utilities that follow an upward trend based on the market following by a downward trend for the next 30 days. To simulate this, we generate $u_{it},\forall~i,t$, from a uniform distribution in $[5,10]$ in an i.i.d. fashion. 
The job size for the $i$-th user, $J_i~\forall~i$, is an i.i.d. integer chosen between $10$ to $100$ (which remains constant throughout the 60 days of the experiment) and $M_t = 5000$ for all $t$. 
To generate an upward market trend, we add $0.5$ to $u_{i,t},\forall~i,t$, with probability 0.55 and $-0.5$ with probability $0.45$ (the probabilities are flipped to generate a downward trend). 
\footnote{Note that the user problem can have multiple solutions. To ensure  convergence to a unique solution with algorithm \ref{algo:price_tracking}, we add a small quadratic regularizer to both SYS-LP and the user problem.}

We compare the following three schemes:
\begin{itemize}
    \item {\bf Optimal pricing}: Here, we assume that the cloud provider is able to solve for optimal resource allocation and maximize system utility by utilizing the knowledge of users' utilities at each day. 
\item {\bf Tracking-based pricing}: In this case, the users' utilities are not known, and the cloud provider tracks users' utilities based on their budget signals (as described in Algorithm \ref{algo:price_tracking}). 
The cloud provider is assumed to update the prices everyday.
\emph{Users send budget signals only once per day.} These budget signals depend on the user's utility function on that day and the prices published by the cloud on that day.

\item {\bf First-come-first-serve} (also known as first-in-first-out): 
Here, the pricing of the resources remains constant (as currently employed by most commercial service providers). 
It is the optimum prices determined by the utilities on day 1 and does not capture the mood of the market.
Each day, the users are allocated resources on a first-come-first-serve (which is assumed to be random in order) or a lottery basis. 
\end{itemize}

\begin{figure}[t]
        \centering
        \includegraphics[scale=0.8]{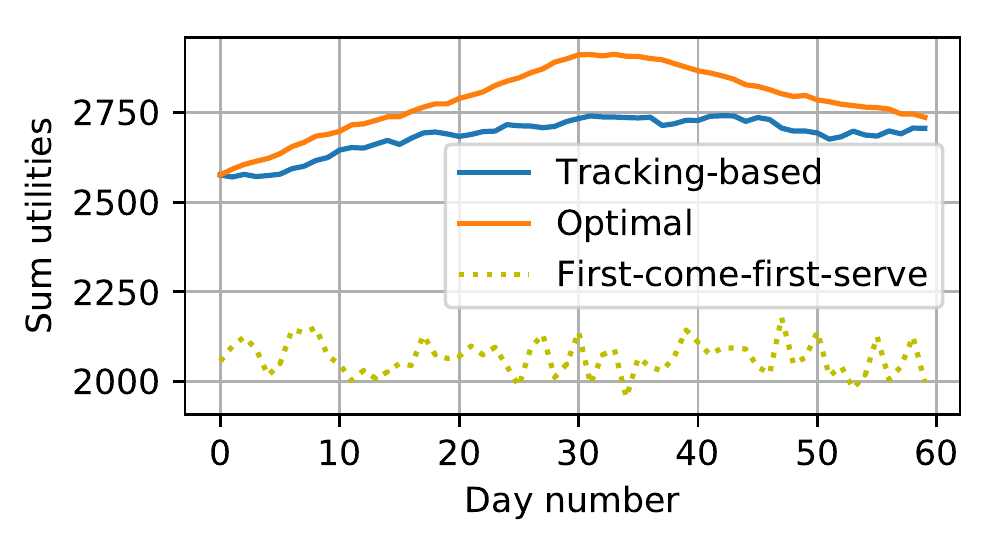}
        \caption{Sum utilities for three resource allocation schemes}
\label{fig:market_sum_utility}
\end{figure}

In Fig. \ref{fig:market_sum_utility}, we plot the sum utility of the system for all the three schemes across 60 days. Note that the tracking-based scheme accurately captures the trends of the market and is always within $8\%$ of the optimal utility, and this happens with only limited feedback from users (where they send budget signals only once per day). The first-come-first-serve scheme is clearly suboptimal with a deviation of as much as $38\%$ from the optimal utility.

\begin{figure*}[t]
    \centering
    \begin{subfigure}[t]{0.45\textwidth}
        \centering
        \includegraphics[scale=0.6]{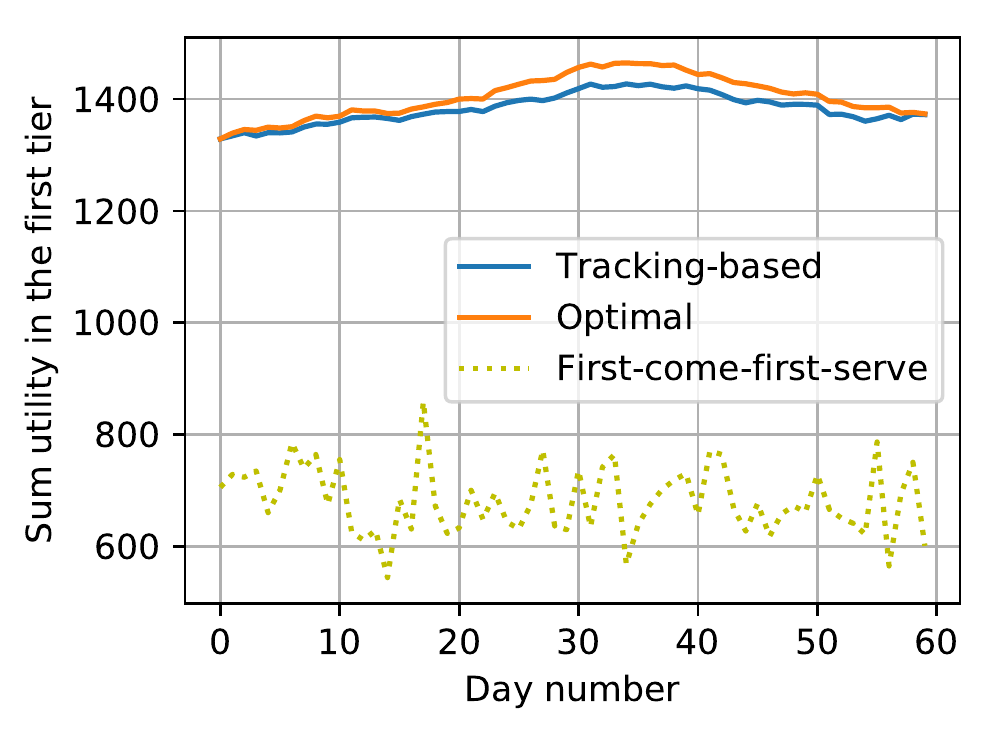}
        \caption{Sum utility in the first tier across 60 days
        }
        \label{fig:utils_tier0}
    \end{subfigure}
    ~
    \begin{subfigure}[t]{0.45\textwidth}
        \centering
        \includegraphics[scale=0.6]{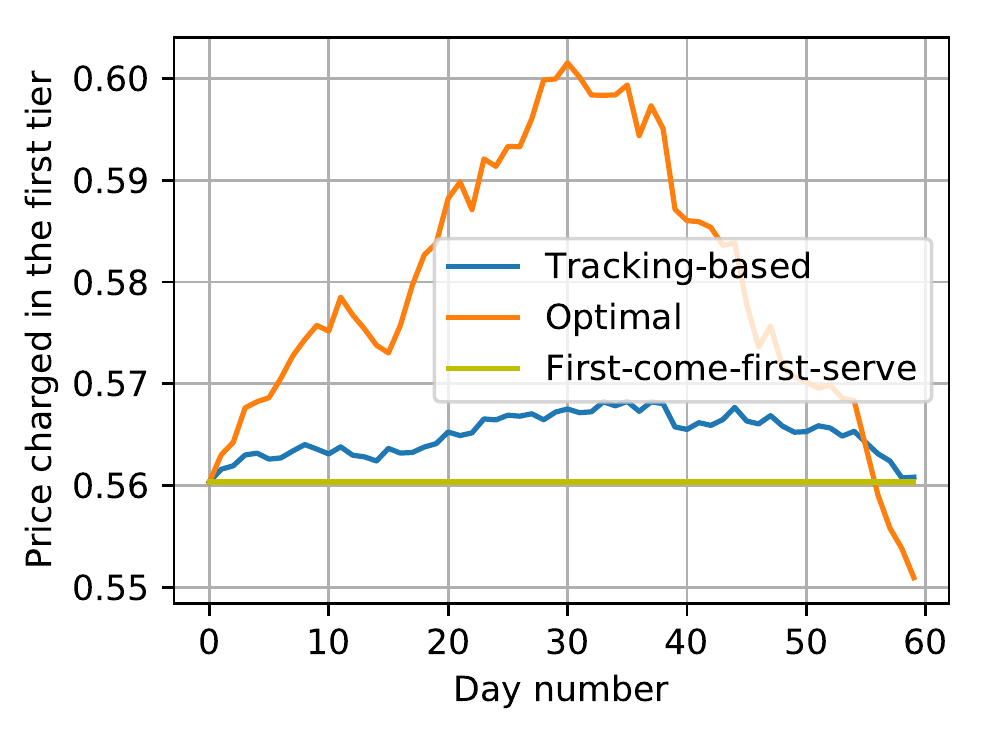}
        \caption{ Corresponding price charged by the cloud provider 
        }
        \label{fig:prices_tier0}
    \end{subfigure}
    \caption{Utility and price charged in the first tier}
\label{fig:market_tier0}
\end{figure*} 

In Fig. \ref{fig:market_tier0}, we plot the utility and price charged across all users only in the first tier. Again, the tracking-based scheme closely follows the optimal utility and lies within $3\%$ of it (Fig. \ref{fig:utils_tier0}). Furthermore, an important advantage of the tracking-based scheme is that it does not drastically change the prices throughout the sixty days (the maximum change is $<2\%$, see Fig. \ref{fig:prices_tier0}). This is unlike the optimal pricing scheme, where the change is $10\%$. From a cloud service provider point of view, the optimal scheme is not ideal since the pricing fluctuates heavily and may alienate users who expect some consistency in the pricing scheme. First-come-first-serve keeps a constant pricing scheme, but it reduces the optimal utility by a factor of 2. Hence, the tracking-based scheme represents an optimal trade-off between utility and fluctuation in pricing (and this trade-off can be controlled by the number of gradient steps in Algorithm \ref{algo:price_tracking}). Moreover, the tracking-based scheme does not require the utility information from users, alleviating potential privacy concerns.

\begin{figure*}[t]
    \centering
    \begin{subfigure}[t]{0.45\textwidth}
        \centering
        \includegraphics[scale=0.6]{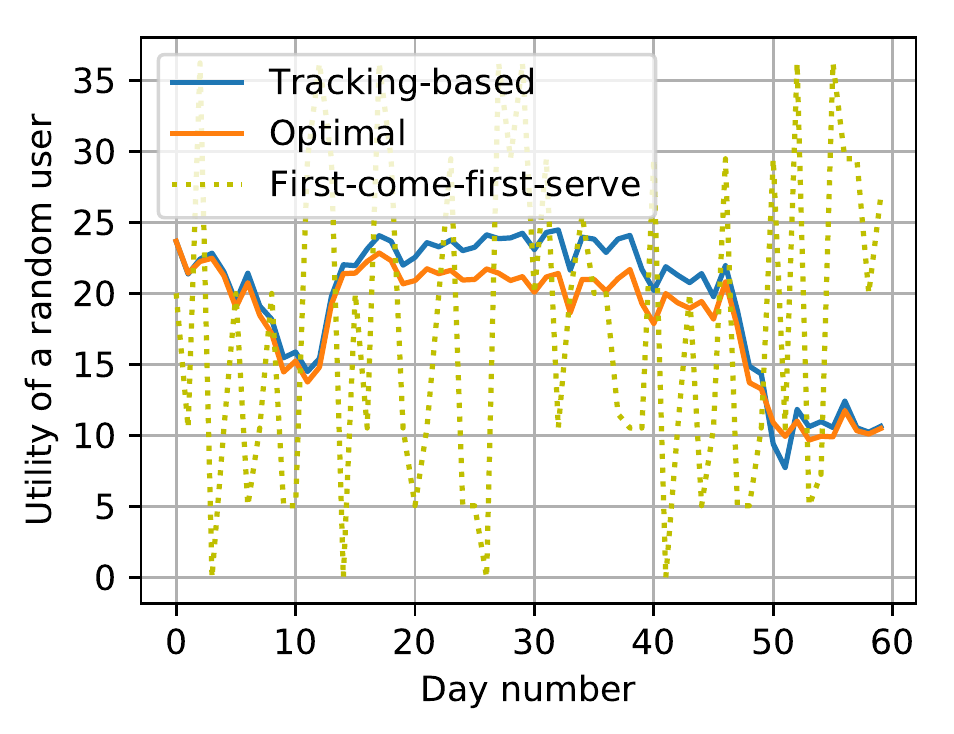}
        \caption{ Utility of a single user over 60 days}
        \label{fig:market_user0_1}
    \end{subfigure}
    ~
    \begin{subfigure}[t]{0.45\textwidth}
        \centering
        \includegraphics[scale=0.6]{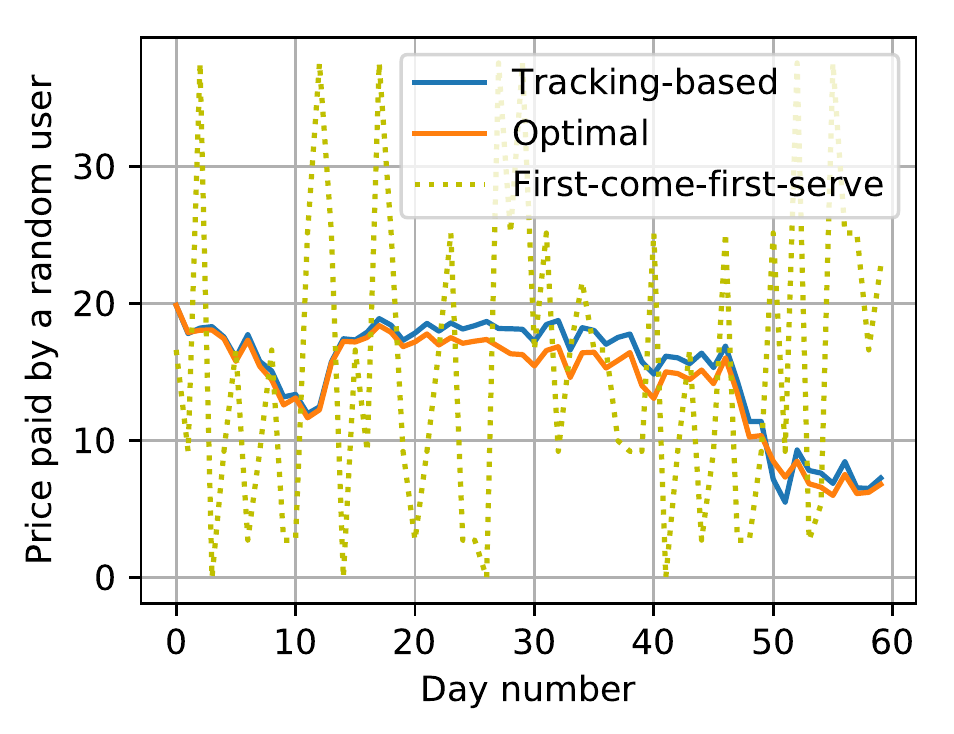}
        \caption{  Corresponding price charged}
        \label{fig:market_user0_2}
    \end{subfigure}
    \caption{Utility and price paid by a randomly chosen user}
\label{fig:market_user0}
\end{figure*} 

In Fig. \ref{fig:market_user0}, we pick a random user and plot her obtained utility and price paid throughout the sixty days. Again, the tracking-based utility is extremely close to the optimal utility of the user. Furthermore, even though the pricing obtained through the tracking-based scheme does not fluctuate much, we see that the revenue obtained from the user is very close to the optimal revenue. This is because by slightly changing the prices according to the mood of the market, the cloud provider is able to nudge the users to update their budget to go close to the optimal budgets. Also note that the price charged to the user on each day (Fig. \ref{fig:market_user0_2}) is less than the utility/willingness-to-pay of that user on that particular day (Fig. \ref{fig:market_user0_1}).


\section{Future work}
\label{sec: conclusion}

Our analysis and simulations show that there are several advantages of using dynamic multi-tier pricing over basic pricing.
This is only the first step towards designing practical game-theoretic resource allocation schemes for the cloud. Below, we describe some limitations of current work, and consequently, potential directions for future research, which are of interest to both industry and academia.

{\bf Utility-based pricing for serverful systems}:
The schemes developed in this paper can be extended to the serverful case provided the assumption that one function requires one machine and unit time holds. 
This assumption is more restrictive for the serverful case since the jobs can be large. 
However, in such scenarios, a large job can be broken into several smaller jobs, each of which requires unit time and machine. 
But this also introduces dependencies between jobs of a user which need to be taken into account in problem formulation. 
We explain such job dependencies in more detail next.

{\bf Job dependencies}: Often, certain jobs are recurrent and require execution of some other specific jobs for their execution.
These dependencies are generally represented as a graph. 
An important future direction is to design improved schedulers that take into account such job dependencies.
Ideally, jobs which have low dependencies or that have more reliable request-for-execution times should get a discount in their pricing because they allow the scheduler to plan more efficiently.

{\bf Wholesale discount}:
Current pricing schemes lack a wholesale discount for customers who are requesting large number of jobs with low variability in job sizes. 
Scheduling such jobs is cheaper for the provider since it can rent entire clusters of machines together, resulting in better resource utilization due to efficient bin packing.   

{\bf Flexible SLAs with probabilistic guarantees}: Cloud service providers define strict Service Level Agreements (SLAs) for their services, e.g.  in AWS Lambda, users are credited $10\%$ of their incurred charges if the error rate of Lambda functions goes beyond $0.05\%$\footnote{For details, see https://aws.amazon.com/lambda/sla/historical/}.
Note that providing such strict and premium SLAs require high-level maintenance of the cloud, the costs of which are indirectly borne by all the users. 
This may be unfair since some users could be equipped to tolerate job failures \cite{rashmi_inference,serverless_straggler_mitigation} and different users could have varying degrees of fault-tolerance.
An interesting future direction is to design pricing schemes that conform with each user's preferences using behavioral preference models from decision theory (see, for example, \cite{phade2019optimal}).

{\bf Heterogeneity in machines and jobs}: 
Currently, we have assumed that all the serverless machines/jobs have the same specifications in terms of their execution time, memory capacity, etc.
However, there is often heterogeneity both in the types of machines and jobs for such systems \cite{ms_azure_2020}, which can be taken into account in the problem formulation.

{\bf Revenue optimal schemes}:
In this paper we have focused on optimizing social welfare.
An alternative criteria would be to optimize the revenue generated by the cloud provider.

{\bf Stability analysis}:
Our algorithm is provably convergent to the optimal allocation and pricing when the conditions are static, and seems to track the optimum when the network conditions vary slowly.
The impact of limited feedback on the ability to track the optimum will be a topic of further studies.
\appendix
\section{Related Work}
\label{sec: related work}

The feature of pricing cloud resources based on delay-sensitivity of users has been lacking from most of the current pricing schemes that are being implemented as well as several real-time dynamic pricing schemes that have been proposed in the literature (we refer the readers to \cite{al2013cloud}, \cite{spot_pricing_survey} and \cite{cloud_pricing_survey2} for surveys on existing and proposed pricing schemes for the cloud). 
Notable exceptions are \cite{wang2012datacenter,joe2012mathematical,halima2017optimal}, where job completion deadlines affect pricing. 
We give a brief survey of related works here.

Our problem formulation is closest to the model in \cite{wang2012datacenter}, where authors consider the problem of revenue maximization in cloud datacenters.
However, we take a different approach towards solving it leading to a widely different pricing and allocation scheme.
(We elaborate on this in Remark~\ref{rem: compare_wang}.)
Further, in \cite{wang2012datacenter}, it is assumed that 
the service provider knows the users' utility functions and each user is charged an amount equal to the utility corresponding to her delay in service. In this paper, we relax the assumption of knowledge of users' utility.

Notice that the two key features of our model are capturing the users' delay-sensitivity and taking a utility-agnostic approach.
In the next couple of paragraphs, we recall some of related works in the literature with respect to these two features.

In relation to the first feature, suppose we ignore the market setting, 
then our problem is closely related to the problem of scheduling jobs in real-time systems.
This problem has been extensively studied in the literature. (See, for example, \cite{jensen1985time,vengerov2005adaptive,wu2004utility,clark1990scheduling}.) 
The notion of utility as a function of completion times has played a key role in many of these works.
Varied algorithms with supporting simulations have been been proposed in these works.
Building on these works, many have proposed resource allocation in the cloud that is based on maximizing the utility of the system. 
For example, in \cite{menasce2006autonomic} and \cite{walsh2004utility}, the authors propose utility-function based approaches for resource allocation in \emph{autonomic computing systems}. 
In \cite{chase2001managing}, the authors use utility functions to allocate resources dynamically and save energy by consolidating load.
In \cite{minarolli2011utility}, the authors propose schemes to dynamically allocate computing resources to virtual machines (such as virtual operating systems) while minimizing operating costs and satisfying QoS constraints. 
This is done by expressing these two goals as a two-tier utility function. Another example is \cite{islam2012empirical}, 
where the authors find the optimal number and size of virtual machines to allocate CPU resources to applications via an automatic resource controller. 
This is done by using a constraint programming approach to maximize the utility accrued. 
Thus, using system utility maximization as the overall objective for designing resource allocation theme has been a recurring theme in the literature.
Our objective in this paper shares this feature with these works.

Pricing in cloud computing is closely related to the second feature, although to the best of our knowledge this idea has not been fully exploited in the literature. 
When cloud computing was first introduced in \cite{weinhardt2009cloud}, the authors claimed that the success of the cloud can only be obtained by developing adequate pricing schemes.  
Since then a variety of pricing schemes for the cloud have been proposed.
Here, we give a small sample of the several works 
that propose and analyze dynamic pricing schemes for various services in the cloud.
In \cite{li2011pricing}, the authors propose an iterative pricing algorithm that uses the historical pricing of resources and determines the final price based on availability of resources for the next round.
In \cite{rohitratana2012impact}, the authors analyze four dynamic pricing schemes and develope an agent-based simulation of a software market.
In \cite{mihailescu2010dynamic}, the authors propose a federated version of dynamic pricing  where the computing resources are being shared by multiple cloud service providers.
Finally, Amazon offers spot pricing, which is another form of dynamic pricing where the resources are priced at a lower rate than fixed pricing but with less guarantee of availability \cite{spot_pricing_survey}.
For further examples of pricing schemes, see \cite{al2013cloud}.

\section{Remarks}
\label{sec: remarks}

\begin{remark}
\label{rem: greedy_sparse}
Corresponding to any non-preemptive scheduling, there exists an implicitly defined priority ordering amongst the users based on their completion times.
Without loss of generality, suppose that this ordering is $1\leftarrow 2 \leftarrow 3 \leftarrow \cdots \leftarrow N$, that is, user $i$ is served along or before user $i+1$ for all $i\in [N-1]$. 
A simple greedy allocation $\tilde \x := (\tilde x_{i,t})_{i \in [N], t \in [T]}$ corresponding to this priority ordering would work as follows: 
Say the cloud has served the first $i-1$ users till time $t$. 
Now, if the number of machines remaining at time $t$ is greater than $J_i$, all of user $i$'s functions are allocated at time $t$. 
Otherwise, the cloud provider continues allocating resources at times $>t$ to user $i$ till her job is complete 
before attending the user $i+1$. 
This process continues till either all users or all of the $T$ tiers are served.

	Corresponding to the greedy allocation with a given ordering, when the cloud provider moves from tier $t$ to $t+1$, there can be at most one user with a partially satisfied job, for all $t \in [T]$. 
Hence, there are at most $T$ instances such that the users jobs are partially allocated, that is, 
$|\{ (i,t) : 0 < x_{i,t} < J_i \}| \leq T$.
Further, there are at most $N$ entries in the matrix $\x$ such that $x_{i,t} = J_i$.
Thus, the allocation matrix $(x_{i,t}), i \in [N], t\in [T]$ can have at most $(N+T)$ non-zero enties, i.e. it is $(N+T)$--sparse.
\end{remark}

\begin{remark}
\label{rem: tighter_bound}
	We observed in Lemma~\ref{lem: sparsity} that the matrix $\x^R$ of allocations is $(N+T)$-sparse. 
	Suppose in addition to it we know that the allocation $\x^R$ allocates partial resources to at most one user in any tier $t$ as is typical in a greedy allocation (see Remark~\ref{rem: greedy_sparse}).
	Then, in the proof of Theorem~\ref{thm:gap_bound} (see Appendix~\ref{sec: proofs}), we have $|S_t| \leq 1$, for all $t$ (where $S_t$ is the set of users that are allocated partial resources in tier $t$), and we can improve the bound to get,
	$$\hat V \geq \left(1 - \frac{\max_i J_i}{\min_t M_t}\right) V^R.$$
\end{remark}

\begin{remark}
\label{rem: NP_hardness}
We note that  
the problems SYS/SYS-ILP are NP-hard in general.
In \cite{du1990minimizing}, the authors consider a job scheduling problem for $N$ jobs, with execution times $\{J_1, \dots, J_N\}$, to be executed on a single machine, where each job has a corresponding due date $d_i$. 
The utility function for each job is assumed to decrease by its \emph{tardiness} defined as the delay in the completion of job $i$ from its due date $d_i$.
Namely, if $T_i$ is the completion time of job $i$, then let $D_i(T_i) := \max\{0, T_i - d_i\}$ be the tardiness.
The authors show that the problem of finding a schedule that minimizes the total tardiness is NP-hard in general (see \cite{du1990minimizing,lawler1977pseudopolynomial,lenstra1977complexity}).
Further, this is proved for the case when the parameters $J_i$ and $d_i$, for all $i$, take integer values.
We note that this is a special case of our problem.
To see this, recall that the parameters in our problem are $N,T,(U_{i,t})_{i \in [N], t \in [T]}, (M_t)_{t \in [T]}$.
Let $N$ be same as the number of different jobs in \cite{du1990minimizing}.
Let $T = \sum_{i=1}^N J_i$.
Think of $\tau_t = t$, for $t \in [T]$.
Thus, each service tier has a unit time interval.
Let the utility functions be $U_i(\tau) = D_i(T+1) - D_i(\tau)$.
Thus, if $T_i$ is the completion tier for player $i$, then her utility is given by $U_{i,{T_i}} = D_i(T + 1) - D_i(T_i),$ for all $i$. 
Let the resource capacities be $M_t = 1$ for all $t \in [T]$.
With this setting, we now observe that the problem of maximizing social welfare is equivalent to minimizing the total tardiness.
Given the result in \cite{du1990minimizing}, we get that our problem is NP-hard, too, in general.
\end{remark}

\begin{remark}
\label{rem: compare_wang}
	In \cite{wang2012datacenter}, the authors arrive at a problem formulation that is similar to SYS.
	They note that the condition \eqref{eq: sys_T_def} is not differentialble, thus making it hard to solve the problem SYS.
	They propose to approximate this condition by the following equation
	\begin{equation}
	\label{eq: useless_transform_approx}
		\hat T_i := \frac{1}{\beta} \log \left(\frac{1}{J_i} \sum_{t = 1}^T e^{\beta t} x_{i,t}\right).
	\end{equation}
	Note that as $\beta \to \infty$, 
	\[
		\hat T_i \to \max\{t \in [T]: x_{i,t} > 0 \}.
	\]
	If we assume that user $i$'s job is completed, i.e. $\sum_{t =1}^T x_{i,t} = J_i$, then $\hat T_i \to T_i$ as $\beta \to \infty$.
	The authors replace $U_{i, T_i}$ with ${U}_i(\hat T_i)$ and use Taylor series first order approximation to further simplify the optimization problem, which can then be solved analytically.
	Our approach here is very different from theirs.
	In the next section, we make an important observation regarding the structure of the optimal solution to SYS-LP (see lemma \ref{lem: sparsity}).
	This observation not only allows us to show that the gap between SYS-LP and SYS is small, but also explains why it is enough to consider our direct relaxation instead of the one like \eqref{eq: useless_transform_approx}.
	Besides, our direct relaxation is useful in decomposing the problem into a cloud problem and several user problems--one for each user-- giving rise to a natural dynamic pricing mechanism (see section \ref{sec:pricing_mechanism}).
	Moreover, in a forthcoming paper, we extend our problem formulation to account for uncertainties in the delay times.
	Our direct relaxation is particularly useful in this extension.
\end{remark}

\begin{remark}
\label{rem: optimal_soln_sparsity}
Consider a scenario where the $P_{it}$'s are sampled independently from non-atomic probability distributions over some finite intervals $[\overline{P_{i,t}}, \underline{P_{i,t}}]$ (where $\overline{P_{i,t}} < \underline{P_{i,t}}$), for all $i, t$.
(A probability distribution over the real numbers is non-atomic if the probability of any single real number occuring is zero.)
In such a scenario, we observe that the probability of the event that a non-trivial relationship amongst the variables $P_{i,t}$'s holds is zero.
Thus, for any choice of $T+1$ instances $(i, t_1, t_2)$ that satisfy \eqref{eq: mu_relation}, we get a non-trivial relationship amongst the variables $P_{i,t}$'s and the probability of this happening is zero.
Since there are finitely many choices for such $T+1$ instances $(i, t_1, t_2)$, we get that the probability of there being more than $T$ partial allocations is zero.
Thus, in a generic case in the sense above, for any optimal solution $x_{i,t}$ to SYS-LP, there can be at most $T$ instances where the users are allocated partial resources, i.e. $0 < x_{i,t} < J_i$. 
Moreover, there are at most $N$ instances where the users get their jobs completed, i.e. $x_{i,t} = J_i$. 
And hence, $x_{i,t}, i\in [N], t\in [T]$ is $(N+T)$-sparse (cf. Remark~\ref{rem: greedy_sparse}). 
\end{remark}

\section{Proofs}
\label{sec: proofs}

\begin{proof}[Proof of Lemma~\ref{lemma:ordering_sys}]
Let $\x$ be any optimal resource allocation to SYS and let $T_i$ the corresponding end times as defined in \eqref{eq: end_time_def}.
Without loss of generality, let the following ordering hold:
\begin{equation}
\label{eq: endtime_ord}
T_1 \leq T_2 \leq \dots \leq T_N.	
\end{equation}
As defined in Remark~\ref{rem: greedy_sparse}, let $\tilde \x$ be the greedy allocation corresponding to the ordering $1\leftarrow 2 \leftarrow 3 \leftarrow \cdots \leftarrow N$.
Let $\tilde T_i$ be the end times corresponding to the allocation $\tilde \x$.
By construction of the greedy allocation, we have
\[
	\tilde T_i = \min \left\{t \in [T] : \sum_{s = 1}^t M_t \geq \sum_{j = 1}^i J_j \right\}, \text{ for all } i \in [N].
\]
On the other hand, because of the ordering \eqref{eq: endtime_ord}, we have
\[
	T_i \geq \min \left\{t \in [T] : \sum_{s = 1}^t M_t \geq \sum_{j = 1}^i J_j \right\}, \text{ for all } i \in [N].
\]
Thus, $\tilde T_i \leq T_i$, for all $i \in [N]$.
Hence $\sum_{i} U_{i,{\tilde T_i}} \geq \sum_{i} U_{i,{T_i}}$ and, thus, $\tilde \x$ is also an optimal allocation.
Since $\tilde \x$ is a non-preemptive scheduling, we have the statement in the lemma.
\end{proof}

\begin{proof}[Proof of Lemma~\ref{lem: sparsity}]
We will first extend the SYS-LP problem by adding a dummy service tier $T+1$ with unlimited capacity $M_{T+1} = \infty$.
Let $U_{i, (T+1)} = 0$ for all users $i \in [N]$.
Let us call this problem SYS-LP-EXT.
For any feasible solution $\x$ to SYS-LP, we can construct a feasible solution $\tilde \x$ to SYS-LP-EXT by executing all the remaining functions in tier $T+1$, i.e. for all $i \in [N], t \in [T+1]$, let
\[
	\tilde x_{i,t} := \begin{cases}
		x_{i,t}, &\text{ if } i \in [N], t \in [T],\\
		J_i - \sum_{t \in [T]} x_{i,t},  &\text{ if } i \in [N], t = T+1.
	\end{cases}
\]
Similarly, we can construct a solution $\x$ to SYS-LP corresponding to any solution $\tilde \x$ to SYS-LP-EXT by restricting it to $i \in [N]$ and $t \in [T]$.
Since $U_{i, (T+1)} = 0$ for all users $i \in [N]$, we have that the objective values for SYS-LP and SYS-LP-EXT match for corresponding feasible solutions.
We will now show that there exists an optimal solution $\tilde x$ to SYS-LP-EXT such that $0 < \tilde x_{i,t} < J_i$ for at most $T$ elements.
Taking the corresponding solution for SYS-LP, we will get the optimal solution with the desired properties for SYS-LP.

Let $\tilde \x$ be an optimal solution to SYS-LP-EXT.
Without loss of generality, let us assume that all users get their jobs completed by the end of tier $T+1$, i.e. $\sum_{t} \tilde x_{i,t} = J_i$, for all $i$.
We have this because tier $T+1$ is assumed to have infinite capacity.
Consider the KKT conditions to problem SYS-LP-EXT similar to \eqref{sys_kkt1}, \eqref{sys_kkt2}, and \eqref{sys_kkt3}, with dual variables $\tilde \mu_{t}, t \in [T+1], \tilde \lambda_i, i \in [N]$.
Since $M_{i, (T+1)} = \infty$, we have $\tilde \mu_{T+1} = 0$.
If $\tilde x_{i,t} > 0$, then $\tilde \mu_t + \tilde \lambda_i = F_{i,t}$.
Suppose user $i$ is allocated partial resource in some tier $t_1 \in [T]$, i.e $0 < \tilde x_{i,t_1} < J_i$.
Then there exists a another tier $t_2 \in [T+1]$ such that $0 <\tilde x_{i, t_2} < J_i$.
We will then have
\[
	\tilde \mu_{t_1} - \tilde \mu_{t_2} = F_{i, t_1} - F_{i, t_2}.
\]
Corresponding to any instance $(i, t_1, t_2)$ consider the equation above.
If we have any $(T+1)$ such distinct instances, then we can eliminate the $\mu_t$ variables and get a non-trivial equality relationship between the variables $F_{i,t}$.
For each collection of $(T+1)$ distinct instances $(i, t_1, t_2)$, we get a non-trivial equality relationship between the variables $F_{i,t}$.
Suppose for the moment that the variables $F_{i,t}$ are such that they do not satisfy any of the equality relationships obtained from $(T+1)$ distinct instances $(i, t_1, t_2)$.
Then, we get that $0 < \tilde x_{i,t} < J_{i,t}$ for at most $T$ elements.
This would give us the required result.
In the rest of the proof, we will extend this argument to any variables $F_{i,t}$.

Consider a neighborhood $\cal N$ in the non-negative orthant of the $N \times (T+1)$-dimensional Euclidean space around the matrix $(F_{i,t})_{i\in[N],t \in [T+1]}$.
Consider the set of points $\cal P$ in this neighborhood such that they do not satisfy any of the equality relationships obtained from $(T+1)$ distinct instances $(i, t_1, t_2)$.
We claim that the set $\cal P$ is dense in the set set $\cal N$.
To see this, note that the set of points in $\cal{N}$ that satisfy any given non-trivial equality relationship is zero. 
Since there are finitely many such equality relationships that we need to consider, we get that the set of points in $\cal N$ that satisfy any of these non-trivial equality relationship is zero. 
Since $\cal P$ is the complement of this set, it is dense in $\cal N$.
As a result, we get that there is a sequence of points $(F_{i,t}^l)_{i\in[N],t \in [T+1]}$, $l \geq 1$, belonging to the set $\cal P$ and converging to $(F_{i,t})_{i\in[N],t \in [T+1]}$.
Since the problem SYS-LP-EXT has a continuous objective function and continuous contraint functions, we get that there exists a sequence of solutions $\tilde \x^l, l \geq 1$ such that $\tilde \x^l$ is an optimal solution to the SYS-LP-EXT problem with variables $(F_{i,t})_{i\in[N],t \in [T+1]}$ replaced by $(F_{i,t}^l)_{i\in[N],t \in [T+1]}$, and $\tilde \x^l$ is convergent.
Let it converge to $\tilde \x'$.
We note that $\tilde \x'$ is an optimal solution to SYS-LP-EXT (see \cite{milgrom2002envelope}).
Then we get that $\tilde \x^l$ has at most $T$ elements such that $0 < \tilde x_{i,t}^l < J_i$.
This implies that $\tilde \x'$ has at most $T$ elements such that $0 < \tilde x_{i,t}' < J_i$.
This completes the proof.
\end{proof}


\begin{proof}[Proof of Theorem~\ref{thm:gap_bound}]
Suppose the resource allocation matrix $\x^R \in \R^{N\times T}$ is such that there are at most $T$ elements such that $0 < x_{i,t} < J_i$.
We know that such an optimal solution exists from Lemma~\ref{lem: sparsity}.
In the optimal resource allocation $\x^R$ for SYS-LP, say a user $i$ is getting non-zero resources in slots $m$ and $n$ (and say $m < n$), i.e. $x_{i,m} > 0, x_{i,n} > 0$ and a user $j$ which is getting resources in slot $m$, i.e. $x_{j,m} > 0$. 
Then, by the optimality of $\x^R$, we have
\begin{equation}\label{relation_P}
F_{j,m} - F_{j,n} \geq F_{i,m} - F_{i,n}.
\end{equation}
We can easily prove the above by redistributing $\epsilon (>0)$ fraction of the job from user $j$ in slot $m$ to user $i$ in slot $m$ (and vice versa in slot $n$). But we know that this can only decrease the objective function in SYS-LP, that is
$$(F_{i,m} - F_{i,n})\epsilon - (F_{j,m} - F_{j,n})\epsilon \leq 0,$$
which proves Eq. \eqref{relation_P}. 

Now, we are ready to prove the theorem by giving an upper bound for the fraction
$$\frac{V^R - \hat V}{V^R},$$ 
where $\hat V$ is the objective of SYS-LP obtained by projecting solution of SYS-LP to integer constraints as described in Eq. \eqref{y_hat}. 
We bound the numerator and denominator at each time slot $t\in [T]$. 
Let $V^R_t$ and $\hat V_t$ be the corresponding utilities obtained only at tier $t\in [T]$, i.e.
\[
	V^R_t := \sum_{i = 1}^N u_{i,t} y^R_{i,t} \quad \text{ and } \quad \hat V_t := \sum_{i = 1}^N u_{i,t} \hat y_{i,t},
\]
where $y^R_{i,t}$ and $\hat y_{i,t}, \forall i, t,$ are as defined in \eqref{eq: y_R_def} and \eqref{y_hat}, respectively.
Let $T_i^R$ be the time at which the job of user $i$ are getting finished according to the resource allocation $\x^R$, that is, 
\[
	T_i^R := \min \{t \in [T] : \sum_{s = 1}^t x^R_{i,s} \geq J_i \}.
\]
Let $T_i^R = T+1$, if $\sum_{t=1}^T x_{i,s}^R < J_i$.
Thus, for tier $t$, we have
\begin{equation}
\frac{V^R_t - \hat V_t}{V^R_t} = \frac{\sum_{i\in S_t}(F_{i,t} - F_{i, T^R_i})x^R_{i,t}}{\sum_{j=1}^N F_{jt}x^R_{j,t}},
\end{equation}   
where $S_t$ is the set of users for which $0 < x^R_{i,t} < J_i$. Formally, $S_t := \{i\in [N], 0 < x^R_{i,t} < J_i\}$. 
We can further write
\begin{align}
\frac{V^R_t - \hat V_t}{V^R_t} &\leq \sum_{i\in S_t}\frac{(F_{i,t} - F_{i, T^R_i})x^R_{i,t}}{\sum_{j=1}^N (F_{j,t} - F_{j, T^R_i}) x^R_{j,t}} \nonumber\\
&\leq \sum_{i\in S_t}\frac{(F_{i,t} - F_{i, T^R_i})x^R_{i,t}}{(F_{i,t} - F_{i, T^R_i})\sum_{j=1}^N x^R_{j,t}},
\end{align}
where the last inequality uses Eq. \eqref{relation_P}. Now, since at time $t$, user $i$ is getting fractional resources, it implies that the system is operating at full capacity, that is $\sum_{j=1}^N x_{j,t} = M_t$. Hence, we get
\begin{align}
\frac{V^R_t - \hat V_t}{V^R_t} \leq \sum_{i\in S}\frac{x^R_{i,t}}{M_t} \leq \sum_{i\in S}\frac{\max_i(J_i)}{M_t},
\end{align}
where the last inequality uses the fact that $x^R_{i,t}\leq \max_i(J_i)$. Also, since there are at most $T$ instances where users are getting partial resources, $|S_t| \leq T$, and, we get
\begin{align}
\frac{V^R_t - \hat V_t}{V^R_t} \leq \frac{T\max_i(J_i)}{M_t} \leq \frac{T\max_i(J_i)}{\min_t M_t}.
\end{align}
Thus,
\begin{align}
V^R - \hat V = \sum_{t=1}^T (V^R_t - \hat V_t)\nonumber \\ 
&\leq \frac{T\max_i(J_i)}{\min_t M_t}\sum_{t=1}^T V^R_t = \frac{T\max_i(J_i)}{\min_t M_t} V^R,
\end{align}
Rearranging, we get
$$\hat V \geq \left(1 - \frac{T(\max_i J_i)}{\min_t M_t}\right) V^R,$$
which proves the desired result.
\end{proof}

\begin{proof}[Proof of Theorem~\ref{thm:decomposition}]
Let $\x$ be an optimal solution to SYS-LP and let $\bmu$ and $\lambda$ be the dual variables corresponding to this solution.
We know that these satisfy the KKT conditions \eqref{sys_kkt1}, \eqref{sys_kkt2}, and \eqref{sys_kkt3}.
Let $\q = \bmu$ and $m_{i,t} = x_{i,t}\mu_{i,t}$ for all $i,t$.

We will now show that $\m_i$ solves USER(i) for this $\q$.
Observe that $m_{i,t} = 0$ if $q_t = 0$ because of the way we have defined $\m$ here.
Thus, it is enough to look at the tiers for which $q_t \neq 0$.
Hence, without loss of generality, we will assume that $q_t \neq 0$ for all $t$. 
Consider the Lagrangian for the user problem USER(i),
\begin{equation}
L(\m_i, p_i) = \sum_{t=1}^T \frac{m_{i,t}}{q_t}\left(F_{i,t} - q_t\right) + p_i\left(J_i - \frac{m_{i,t}}{q_t}\right), \forall i \in [N],
\end{equation}
 where $p_i$ is the dual variable corresponding to the job size constraint \eqref{user_prob_cons} in the user problem. 
 The KKT conditions can, thus, be written as 
\begin{align}
\label{user_kkt1} \mu_t + p_i 
&\begin{cases}
= F_{i,t}, & \text{ if } m_{i,t} > 0\\
\geq F_{i,t}, & \text{ if } m_{i,t} = 0,
\end{cases}
\quad \quad \forall i, t,
\\
\label{user_kkt2}
\sum_{t=1}^{T} \frac{m_{i,t}}{q_t} 
&\begin{cases}
= J_i, & \text{ if } p_{i} > 0\\
\leq J_i, & \text{ if } p_{i} = 0,
\end{cases}
\quad \quad~~~~\forall t.
\end{align}
Taking $p_i = \lambda_i$, we get that these KKT conditions are satisfied.
Thus, $\m_i$ is an optimal solution to USER(i) with $\q = \bmu$.

Now we will show that $\m$ is an optimal solution for the CLOUD problem.
The Lagrangian for the CLOUD problem is given by
\begin{equation}\label{lag:cloud}
L(\m, \tilde \q) = \sum_{t=1}^T\sum_{i=1}^N m_{i,t}\log x_{i,t} + \sum_{t=1}^T q_t\left(M_t - \sum_{i=1}^N x_{i,t}\right),
\end{equation}
where $\tilde \q = (\tilde q_t, t\in [N])$ is the dual variable corresponding to the constraint \eqref{cloud_prob_cons} in the CLOUD problem. 
Let $\tilde q = \mu$.
If $m_{i,t} > 0$, then $x_{i,t} > 0$, and differentiating the Lagrangian with respect to $x_{i,t}$ we get
\[
	\frac{\partial L(\m, \q)}{\partial x_{i,t}} =  \frac{m_{i,t}}{x_{i,t}} - q_t = 0.
\]
If $m_{i,t} = 0$, then $\frac{\partial L(\m, \tilde \q)}{\partial x_{i,t}} \leq 0$, since $q_t \geq 0$.
Further, from \eqref{sys_kkt2}, we have
\begin{align}
\label{cloud_kkt2}
\sum_{i=1}^{N} x_{it} 
&\begin{cases}
= M_t, & \text{ if } q_{t} > 0\\
\leq M_t, & \text{ if } q_{t} = 0,
\end{cases}~~\forall t.
\end{align}
Thus, $\x$ and $\q$ satisfy the KKT conditions for the CLOUD problem. Hence $\x$ is an optimal solution to CLOUD with $\m$.

Thus we have showed statements (i) and (ii) in Theorem~\ref{thm:decomposition}. Statement (iii) follows from construction and statement (iv) follows from \eqref{sys_kkt2}.
%
We now prove the later assertion, namely, if we have an equilibrium solution $\x, \m, \q$ that satisfy (i), (ii), (iii), and (iv), then $\x$ solves the system problem SYS-LP.
To see this, let $\x, \m, \q$ be such an equilibrium solution.
Take $\bmu = \q$.
Since $\m_i$ is an optimal solution to USER(i) with $\q$, there exists dual a variable $p_i$ corresponding to the constraint \eqref{user_prob_cons}.
Take $\lambda_i = p_i$ for all $i$.
It is easy to check that $\x, \bmu, \bl$ satisfy the KKT consitions \eqref{sys_kkt1}, \eqref{sys_kkt2}, and \eqref{sys_kkt3}, and hence, form an optimal solution to SYS-LP.
This completes the proof of the theorem.
\end{proof}

\bibliographystyle{IEEEtran} 
\bibliography{bibli}

\end{document}